\def\year{2016}
\documentclass[letterpaper]{article}
\usepackage{aaai16}
\usepackage{times}
\usepackage{comment}
\usepackage{helvet}
\usepackage{courier}
\usepackage{multirow}
\newcommand{\namecite}[1]{\citeauthor{#1}~\shortcite{#1}}
\newcommand{\citenobrackets}[1]{\citeauthor{#1}~\citeyear{#1}}
\newcommand\citet{\namecite}

\usepackage{graphicx}
\usepackage{subfigure}
\usepackage{booktabs}
\usepackage{color}
\usepackage{algorithm}
\usepackage{algorithmic}
\usepackage{placeins}

\usepackage{amsmath}
\usepackage{amsthm}
\usepackage{amssymb}
\theoremstyle{definition}   
\newtheorem{mydef}{Definition}
\newtheorem{remark}{Remark}
\theoremstyle{plain}   
\newtheorem{proposition}{Proposition}
\newtheorem{theorem}{Theorem}
\newtheorem{lemma}{Lemma}
\newtheorem{corollary}{Corollary}

\newcommand{\Exp}{\mathbb{E}}

\newcommand{\ashish}[1]{{\color{blue} [AS: {#1}]}}
\newcommand{\stefano}[1]{{\color{green} [SE: {#1}]}}

\newcommand{\js}[1]{}

\newcommand{\half}{{\frac{1}{2}}}
\newcommand{\limn}{\lim_{n \to \infty}}

\newcommand{\kbound}{{3.6 - \frac{5}{4}\log_2{\alpha}}}

\newcommand{\bigOh}{\mathrm{O}}

\newcommand{\I}{\mathbb{I}}
\newcommand\Tstrut{\rule{0pt}{2.6ex}}
\newcommand\minf{\tilde{f}}
\newcommand\numberthis{\refstepcounter{equation}\tag{\theequation}} 

\title{Closing the Gap Between Short and Long XORs for Model Counting}

\author{Shengjia Zhao \\  Computer Science Department 
\\ Tsinghua University \\ zhaosj12@mails.tsinghua.edu.cn
\And 
Sorathan Chaturapruek \\  Computer Science Department 
\\ Stanford University \\ sorathan@cs.stanford.edu 
\And Ashish Sabharwal  
\\ Allen Institute for AI \\ Seattle, WA \\ ashishs@allenai.org
\And
Stefano Ermon \\  Computer Science Department 
\\ Stanford University \\ ermon@cs.stanford.edu}

\begin{document}
 
\maketitle

\begin{abstract}
Many recent algorithms for approximate model counting are based on a reduction to combinatorial searches over random subsets of the space defined by parity or XOR constraints. 
Long parity constraints (involving many variables) provide strong theoretical guarantees but are computationally difficult. Short parity constraints are easier to solve but have weaker statistical properties. It is currently not known how long these parity constraints need to be. We close the gap by providing matching necessary and sufficient conditions on the required asymptotic length of the parity constraints. Further, we provide a new family of lower bounds and the first non-trivial upper bounds on the model count that are valid for arbitrarily short XORs. We empirically demonstrate the effectiveness of these bounds on model counting benchmarks and in a Satisfiability Modulo Theory (SMT) application motivated by the analysis of contingency tables in statistics.	
\end{abstract}

\section{Introduction}

\js{can probably cut first paragraph.  I think it would be better if you define your task of interest up-front.
i don't know what a contingency table / chi-squared test is, e.g. like, i'd roughly start with what's currently paragraph 3
except before that, set up the task talk about why it's hard
and how you'll solve it. citations are good here, too
you don't need to talk about probability being good / probabilistic inference being hard, because everyone at ICML already agrees with that (although it being hard for this particular task is worth discussing)}

Model counting is the problem of computing the number of distinct solutions of a given Boolean formula. It is a classical problem that has received considerable attention from a theoretical point of view~\cite{valiant1979complexity,stockmeyer1985approximation}, as well as from a practical perspective~\cite{sang2004combining,gogate2007approximate}. Numerous probabilistic inference and decision making tasks, in fact, can be directly translated to (weighted) model counting problems~\cite{richardson2006markov,gogate2011probabilistic}. As a generalization of satisfiability testing, the problem is clearly intractable in the worst case. Nevertheless, there has been considerable success in both exact and approximate model counting algorithms, motivated by a number of applications~\cite{sang2005solving}. 

Recently, approximate model counting techniques based on randomized hashing have emerged as one of the leading approaches~\cite{mbound,cav13,ermon2014low,ivrii2015computing,achlioptasstochastic,belle2015hashing}. While approximate, these techniques provide strong guarantees on the accuracy of the  results in a probabilistic sense. Further, these methods all reduce model counting to a small number of combinatorial searches on a randomly projected version of the original formula, obtained by augmenting it with randomly generated parity or XOR constraints. This approach allows one to leverage decades of research and engineering in combinatorial reasoning technology, such as fast satisfiability (SAT) and SMT solvers~\cite{biere2009handbook}. 

While modern solvers have witnessed tremendous progress over the past 25 years, model counting techniques based on hashing tend to produce instances that are difficult to solve. In order to achieve strong (probabilistic) accuracy guarantees, existing techniques require each randomly generated parity constraint to be relatively long, involving roughly half of the variables in the original problem. Such constraints, while easily solved in isolation using Gaussian Elimination, are notoriously difficult to handle when conjoined with the original formula~\cite{ghss07:shortxors,ermon2014low,ivrii2015computing,achlioptasstochastic}. Shorter parity constraints, i.e., those involving a relative few variables, are friendlier to SAT solvers, but their statistical properties are not well understood.

\citet{ermon2014low} showed that long parity constraints are not strictly necessary, and that one can obtain the \emph{same accuracy guarantees} 
using shorter XORs, which are computationally much more friendly. They provided a closed form expression, allowing an easy computation of an XOR length that suffices, given various parameters such as the number of problem variables, the number of constraints being added, and the size of the solution space under consideration. It is, however, currently not known how tight their sufficiency condition is, how it scales with various parameters, or whether it is in fact a necessary condition.

\emph{We resolve these open questions by providing an analysis of the optimal asymptotic constraint length} required for obtaining high-confidence approximations to the model count. Surprisingly, for formulas with $n$ variables, we find that when $\Theta(n)$ constraints are added, a constraint length of $\Theta(\log n)$ is both necessary and sufficient. This is a significant improvement over standard long XORs, which have length $\Theta(n)$. Constraints of logarithmic length can, for instance, be encoded efficiently with a polynomial number of clauses. 
We also study upper bounds on the minimum sufficient constraint length, which evolve from $\bigOh(\log n)$ to $\bigOh(n^\gamma \log^2 n)$ to $n/2$ across various regimes of the number of parity constraints.

As a byproduct of our analysis, we obtain a \emph{new family of probabilistic upper and lower bounds that are valid regardless of the constraint length} used. These upper and lower bounds on the model count reach within a constant factor of each other as the constraint density approaches the aforementioned optimal value. The bounds gracefully degrade as we reduce the constraint length and the corresponding computational budget. While lower bounds for arbitrary XOR lengths were previously known~\cite{ghss07:shortxors,uai13LPCount}, the upper bound we prove in this paper is the first non-trivial upper bound in this setting. Remarkably, even though we rely on random projections and therefore only look at subsets of the entire space (a \emph{local} view, akin to traditional sampling), we are able to say something about the \emph{global} nature of the space, i.e., a probabilistic upper bound on the number of solutions.

We evaluate these new bounds on standard model counting benchmarks and on a new counting application arising from the analysis of contingency tables in statistics. These data sets are common in many scientific domains, from sociological studies to ecology~\cite{sheldon2011collective}. We provide a new approach based on SMT solvers and a bit-vector arithmetic encoding. Our approach scales very well and produces accurate results on a wide range of benchmarks. It can also handle additional constraints on the tables, which are very common in scientific data analysis problems, where prior domain knowledge translates into constraints on the tables (e.g., certain entries must be zero because the corresponding event is known to be impossible). We demonstrate the capability to handle structural zeroes~\cite{chen2007conditional} in real experimental data.

\section{Preliminaries: Counting by Hashing}

Let $x_1, \cdots, x_n$ be $n$ Boolean variables. Let $S \subseteq \{0,1\}^n$ be a large, high-dimensional set\footnote{We restrict ourselves to the binary case for the ease of exposition. Our work can be naturally extended to categorical variables.}. We are interested in computing $|S|$, the number of elements in $S$, when $S$ is defined succinctly through conditions or constraints that the elements of $S$ satisfy and membership in $S$ can be tested using an NP oracle. For example, when $S$ is the set of solutions of a Boolean formula over $n$ binary variables, the problem of computing $|S|$ is known as model counting, which is the canonical $\#$-P complete problem~\cite{valiant1979complexity}.

In the past few years, there has been growing interest in approximate probabilistic algorithms for model counting. It has been shown \cite{mbound,wishicml13,cav13,achlioptasstochastic,belle2015hashing} that one can reliably estimate $|S|$, both in theory and in practice, by repeating the following simple process: randomly partition $S$ into $2^m$ cells and select one of these lower-dimensional cells, and compute whether $S$ has at least 1 element in this cell (this can be accomplished with a query to an NP oracle, e.g., invoking a SAT solver). Somewhat surprisingly, repeating this procedure a small number of times provides a constant factor approximation to $|S|$ with high probability, even though counting problems (in $\#$-P) are believed to be significantly harder than decision problems (in NP).

The correctness of the approach relies crucially on how the space is randomly partitioned into cells. 
All existing approaches partition the space into cells using parity or XOR constraints. A parity constraint defined on a subset of variables checks whether an odd or even number of the variables take the value 1. Specifically, $m$ parity (or XOR) constraints are generated, and $S$ is partitioned into $2^m$ equivalence classes based on which parity constraints are satisfied.

The way in which these constraints are generated affects the quality of the approximation of $|S|$ (the model count) obtained. Most methods \emph{randomly} generate $m$ parity constraints by adding each variable to each constraint with probability $f \leq 1/2$. This construction can also be interpreted as defining a hash function, mapping the space $\{0,1\}^n$ into $2^m$ hash bins (cells). Formally,
\begin{mydef}
\label{def:sparse-hash}
Let $A\in \{0,1\}^{m \times n}$ be a random matrix whose entries are Bernoulli i.i.d.\ random variables of parameter $f \leq 1/2$, i.e., $\Pr[A_{ij}=1] = f$. Let $b \in \{0,1\}^m$ be chosen uniformly at random, independently from $A$. Then, $\mathcal{H}^f_{m \times n}=\{h_{A,b}: \{0,1\}^n \rightarrow \{0,1\}^m\}$, where $h_{A,b}(x)= A x + b \mod{2}$ and $ h_{A,b} \in_R \mathcal{H}_{m \times n}^f$ is chosen randomly according to this process, is a family of \emph{$f$-sparse hash functions}.
\end{mydef}

The idea to estimate $|S|$ is to define progressively smaller cells (by increasing $m$, the number of parity constraints used to define $h$), until the cells become so small that no element of $S$ can be found inside a (randomly) chosen cell. The intuition is that the larger $|S|$ is, the smaller the cells will have to be, and we can use this information to estimate $|S|$.

Based on this intuition, we give a hashing-based counting procedure (Algorithm~\ref{algo_wish}, SPARSE-COUNT), which relies on an NP oracle $\mathcal{O}_S$ to check whether $S$ has an element in the cell. It is adapted from the SPARSE-WISH algorithm of \citet{ermon2014low}. The algorithm takes as input $n$ families of \emph{$f$-sparse hash functions} $\{\mathcal{H}^{f_i}_{i \times n}\}_{i=0}^n$, used to partition the space into cells. In practice, line \ref{checkwithoracle} is implemented using a SAT solver as an NP-oracle. In a model counting application, this is accomplished by adding to the original formula $i$ parity constraints generated as in Definition \ref{def:sparse-hash} and checking the satisfiability of the augmented formula. 

\begin{algorithm}[tb]
\caption{SPARSE-COUNT $(\mathcal{O}_S, \Delta, \alpha, \{\mathcal{H}^{f_i}_{i \times n}\}_{i=0}^n )$ 
}
\label{algo_wish}
\begin{algorithmic}[1]
\STATE $T \leftarrow \left\lceil \frac{\log \left(1/\Delta\right)}{\alpha} \log n \right\rceil$
\STATE $i=0$
\WHILE {$i \leq n$}
  \FOR {$t=1, \cdots, T$}
    \STATE $h^i_{A,b} \leftarrow$ hash function sampled from $\mathcal{H}^{f_i}_{i \times n}$
    \STATE Let $S(h^i_{A,b}) = |\{x \in S \mid h^i_{A,b}(x) = 0\}|$
    \STATE $w_i^t \leftarrow \mathbb{I}[S(h^i_{A,b}) \geq 1]$, invoking $\mathcal{O}_S$ \label{checkwithoracle}
  \ENDFOR
  \IF {$\mathrm{Median}(w_i^1, \cdots, w_i^T) < 1$}
    \STATE $\mathrm{break}$
  \ENDIF
  \STATE $i = i +1$
\ENDWHILE
\STATE Return $\lfloor 2^{i-1} \rfloor$
\end{algorithmic}
\end{algorithm}

Typically, $\{\mathcal{H}^{\frac{1}{2}}_{i \times n}\}$ is used, corresponding to XORs where each variable is added with probability $1/2$ (hence with average length $n/2$). We call these \textbf{long parity constraints}. In this case, it can be shown that SPARSE-COUNT will output a factor $16$ approximation of $|S|$ with probability at least $1-\Delta$~\cite{ermon2014low}. Unfortunately, checking satisfiability (i.e., $S(h^i_{A,b}) \geq 1$, line \ref{checkwithoracle}) has been observed to be very difficult when many long parity constraints are added to a formula~\cite{ghss07:shortxors,ermon2014low,ivrii2015computing,achlioptasstochastic}. Note, for instance, that while a parity constraint of length one simply freezes a variable right away, a parity constraint of length $k$ can be propagated only after $k-1$ variables have been set. From a theoretical perspective, parity constraints are known to be fundamentally difficult for the resolution proof system underlying SAT solvers (cf.\ exponential scaling of Tseitin tautologies~\cite{tseitin68}).
A natural question, therefore, is whether \textbf{short parity constraints} can be used in SPARSE-COUNT and provide reliable bounds for $|S|$.

Intuitively, for the method to work we want the hash functions $\{\mathcal{H}^{f_i}_{i \times n}\}$ to have a small collision probability. In other words, we want to ensure that when we partition the space into cells, configurations from $S$ are divided into cells evenly. This gives a direct relationship between the original number of solutions $|S|$ and the (random) number of solutions in one (randomly chosen) cell, $S(h)$. More precisely, we say that the hash family \emph{shatters} $S$ if the following holds:
\begin{mydef}
\label{def:shatter}
For $\epsilon > 0$, a family of hash functions $\mathcal{H}^{f}_{i \times n}$ $\epsilon$-\emph{shatters} a set $S$ if $\Pr[S(h) \geq 1] \geq 1/2+\epsilon$ when $h \in_R \mathcal{H}^{f}_{i \times n}$, where $S(h) = |\{x \in S \mid h(x) = 0\}|$.
\end{mydef}
The crucial property we need to obtain reliable estimates is that the hash functions (equivalently, parity constraints) are able to shatter sets $S$ with arbitrary ``shape''. This property is both sufficient and necessary for SPARSE-COUNT to provide accurate model counts with high probability:
\begin{theorem}
\label{thm:const}
(Informal statement) A necessary and sufficient condition for SPARSE-COUNT to provide a constant factor approximation to $|S|$ is that each family $\mathcal{H}^{f_i}_{i \times n}$ $\epsilon$-shatters all sets $S'$ of size $|S'|=2^{i + c}$ for some $c\geq 2$.
\end{theorem}
A formal statement, along with all proofs, is provided in a companion technical report~\cite{zhaotech}.

Long parity constraints, i.e., $1/2$-sparse hash functions, are capable of shattering sets of arbitrary shape. When $h \in_R \mathcal{H}^{\frac{1}{2}}_{i \times n}$, it can be shown that configurations $x\in \{0,1\}^n$ are placed into hash bins (cells) pairwise independently, and this guarantees shattering of sufficiently large sets of arbitrary shape.   Recently, \citet{ermon2014low} showed that sparser hash functions can be used for approximate counting as well. In particular, $f^*$-sparse hash functions, for sufficiently large $f^* \lneqq 1/2$, were shown to have good enough shattering capabilities.
It is currently not known whether $f^*$ is the optimal constraint density.

\section{Asymptotically Optimal Constraint Density}

We analyze the asymptotic behavior of the minimum constraint density $f$ needed for SPARSE-COUNT to produce correct bounds with high confidence. As noted earlier, the bottleneck lies in ensuring that $f$ is large enough for a randomly chosen hash bin to receive at least one element of the set $S$ under consideration, i.e., the hash family shatters $S$.

\begin{mydef}
Let $n, m \in \mathbb{N}, n \geq m$. For any fixed $\epsilon > 0$, the \emph{minimum constraint density} $f = \minf_\epsilon(m,n)$ is defined as the pointwise smallest function such that for any constant $c \geq 2$, $\mathcal{H}^f_{m \times n}$ $\epsilon$-shatters all sets $S \in \{0,1\}^n$ of size $2^{m+c}$.
\end{mydef}

We will show (Theorem~\ref{thm:asymp}) that for any $\epsilon > 0$, $\minf_\epsilon(m,n) = \Omega(\frac{\log m}{m})$, and this is asymptotically tight when $\epsilon$ is small enough and $m = \Theta(n)$, which in practice is often the computationally hardest regime of $m$. Further, for the regime of $m = \Theta(n^\beta)$ for $\beta < 1$, we show that $\minf_\epsilon(m,n) = \bigOh(\frac{\log^2 m}{m})$. Combined with the observation that $\minf_\epsilon(m,n) = \Theta(1)$ when $m = \Theta(1)$, our results thus reveal how the minimum constraint density evolves from a constant to $\Theta(\frac{\log m}{m})$ as $m$ increases from a constant to being linearly related to $n$.

The \emph{minimum average constraint length}, $n \cdot \minf_\epsilon(m,n)$, correspondingly decreases from $n/2$ to $\bigOh(n^{1-\beta} \log^2 n)$ to $\Theta(\log n)$, showing that in the computationally hardest regime of $m = \Theta(n)$, the parity constraints can in fact be represented using only $2^{\Theta(\log n)}$, i.e., a polynomial number of SAT clauses. 


\begin{theorem}
\label{thm:asymp}
Let $n, m \in \mathbb{N}, n \geq m,$ and $\kappa > 1$. The minimum constraint density, $\minf_\epsilon(m,n)$, behaves as follows:
\begin{enumerate}

\item Let $\epsilon > 0$. There exists $M_\kappa$ such that for all $m \geq M_\kappa$:
\[ \minf_\epsilon(m,n) > \frac{\log{m}}{\kappa\, m}\]

\item Let $\epsilon \in (0,\frac{3}{10}), \alpha \in (0,1),$ and $m = \alpha n$. There exists $N$ such that for all $n \geq N$:
\[ \minf_\epsilon(m,n) \leq \left(\kbound\right) \frac{\log{m}}{m} \]

\item Let $\epsilon \in (0,\frac{3}{10}), \alpha, \beta \in (0,1),$ and $m = \alpha n^\beta$. There exists $N_\kappa$ such that for all $n \geq N_\kappa$:
\[ \minf_\epsilon(m,n) \leq \frac{\kappa\, (1-\beta)}{2\beta}\frac{\log^2{m}}{m} \]

\end{enumerate}
\end{theorem}

The lower bound in Theorem~\ref{thm:asymp} follows from analyzing the shattering probability of an $m+c$ dimensional hypercube $S_c=\{x \mid x_j = 0 \ \ \forall j > m+c\}$. Intuitively, random parity constraints of density smaller than $\frac{\log m}{m}$ do not even touch the $m+c$ relevant (i.e., non-fixed) dimensions of $S_c$ with a high enough probability, and thus cannot shatter $S_c$ (because all elements of $S_c$ would be mapped to the same hash bin).

For the upper bounds, we exploit the fact that $\minf(m,n)$ is at most the $f^*$ function introduced by \namecite{ermon2014low} and provide an upper bound on the latter. Intuitively, $f^*$ was defined as the minimum function such that the variance of $S(h)$ is relatively small. The variance was upper bounded by considering the worst case ``shape'' of $|S|$: points packed together unrealistically tightly, all fitting together within Hamming distance $w^*$ of a point. For the case of $m = \alpha n$, we observe that $w^*$ must grow as $\Theta(n)$, and divide the expression bounding the variance into two parts: terms corresponding to points that are relatively close (within distance $\lambda n$ for a particular $\lambda$) are shown to contribute a vanishingly small amount to the variance, while terms corresponding to points that are farther apart are shown to behave as if they contribute to $S(h)$ in a pairwise independent fashion. The $\frac{\log m}{m}$ bound is somewhat natural and also arises in the analysis of the rank of sparse random matrices and random sparse linear systems~\cite{kolchin1999random}. For example, this threshold governs the asymptotic probability that a matrix $A$ generated as in Definition \ref{def:sparse-hash} has full rank~\cite{cooper2000rank}. The connection arises because, in our setting, the rank of the matrix $A$ affects the quality of hashing. For example, an all-zero matrix $A$ (of rank $0$) would map all points to the same hash bucket.  


\section{Improved Bounds on the Model Count}

In the previous sections, we established the optimal (smallest) constraint density that provides a constant factor approximation on the model count $|S|$. However, depending on the size and structure of $S$, even adding constraints of density $f^\ast\ll 0.5$ can lead to instances that cannot be solved by modern SAT solvers (see Table \ref{table:SATCount} below). 

In this section we show that for $f < f^\ast$ we can still obtain probabilistic upper and lower bounds. The bounds constitute a trade off between small $f$ for easily solved NP queries and $f$ close to $f^*$ for guaranteed constant factor approximation.

To facilitate discussion, we define $S(h) = |\{x \in S \mid h(x) = 0\}| = |S \cap h^{-1}(0)|$ to be the random variable indicating how many elements of $S$ \emph{survive} $h$, when $h$ is randomly chosen from $\mathcal{H}^f_{m \times n}$ as in Definition \ref{def:sparse-hash}.
Let $\mu_S = \Exp[S(h)]$ and $\sigma^2_S = \mathrm{Var}[S(h)]$. Then, it is easy to verify that irrespective of the value of $f$, $\mu_S = |S|2^{-m}$. $\mathrm{Var}[S(h)]$ and $\Pr[S(h) \geq 1]$, however, do depend on $f$.

\subsection{Tighter Lower Bound on $|S|$}

Our lower bound is based on Markov's inequality and the fact that the mean of $S(h)$, $\mu_S = |S|2^{-m}$, has a simple linear relationship to $|S|$. 
Previous probabilistic lower bounds~\cite{ghss07:shortxors,uai13LPCount} are based on the following observation: it is very unlikely for at least half of $T$ repetitions of applying $h$ to $S$ to result in some element of $S$ surviving unless there are at least $2^{m-2}$ solutions in $S$. 
Otherwise, $\mu_S$ would be too small (specifically, $\leq 1/4$), making it unlikely for solutions to survive often.

Unlike previous methods, we not only check whether the estimated $\Pr[S(h) \geq 1]$ is at least $1/2$, but also consider an empirical estimate of $\Pr[S(h) \geq 1]$. This results in a tighter lower bound, with a probabilistic correctness guarantee derived using Chernoff's bound.

\begin{lemma}
\label{lemma:lb_ideal}
Let $S \subseteq \{0,1\}^n, f \in [0,1/2],$ and for each $m \in \{1, 2, \ldots, n\}$, let $h_m \in_R \mathcal{H}^f_{m \times n}$. Then,
\begin{align}
\left|S\right| & \geq \max_{m=1}^n \ 2^m \Pr[S(h_m) \geq 1].
\end{align}
\end{lemma}

Our theoretical lower bound is $L=2^{m}\mathbb \Pr\left[S(h)\geq1\right]$, which satisfies $|S| \geq L$ by the previous Lemma. 
In practice, we cannot compute $\Pr[S(h)\geq 1]$ exactly.
so our practical lower bound $\hat{L}$ will be based on an empirical estimate
$\Pr_\text{est}\left[S(h)\geq1\right]$ derived from samples. Because $\hat{L}$ is a random variable,
we would like to have a statement of the form $\Pr[\left|S\right|\geq \hat{L}] \geq 1-\delta$, where the probability is with respect to $\hat{L}$. This is formalized by the following Theorem.

\begin{theorem}
\label{thm:tighter-LB}
Let $m>0, T>0,$ and $h_{m}^1, \cdots, h_{m}^T$ be hash functions sampled independently from $\mathcal{H}^f_{m \times n}$. Let $Y_{k}=\mathbb I\left[S(h_{m}^k)\geq 1\right], Y = \sum_{k=1}^{T}Y_{k},$ and $\Pr_{\text{est}}\left[S(h)\geq1\right] = Y/T$ be random variables.
Let $\kappa > 0$, $c>0$. 
Define a random variable $\mathcal{B} = \mathcal{B}(S,h_{m}^1, \cdots, h_{m}^T)$ as $\frac{2^{m} c}{(1+\kappa)}$ if
$
\Pr_{\text{est}}\left[S(h)\geq1\right] \geq c
$
and $0$ otherwise. Then
\begin{align*}
\Pr\left[|S|\geq\mathcal{B}\right] \geq 1 - \exp\left(-\frac{\kappa^2 cT}{(1+\kappa)(2+\kappa)} \right).
\end{align*}
\end{theorem}

\subsection{New Upper Bound for $|S|$}

The upper bound expression for $f$ that we derive next is based on the contrapositive of the observation of \namecite{ermon2014low} that the larger $|S|$ is, the smaller an $f$ suffices to shatter it.

For $n, m, f,$ and $\epsilon(n,m,q,f)$ from \cite{ermon2014low}, define:
\begin{align}
\label{def:vq}
v(q)=\frac{q}{2^{m}}\left(1+\epsilon(n, m, q, f)\cdot(q-1)-\frac{q}{2^{m}}\right)
\end{align}
This quantity is an upper bound on the variance $\mathrm{Var}[S(h)]$ of the number of surviving solutions as a function of the size of the set $q$, the number of constraints used $m$, and the statistical quality of the hash functions, which is controlled by the constraint density $f$. The following Lemma characterizes the asymptotic behavior of this upper bound on the variance:
\begin{lemma}
\label{lemma:inc}
$q^2/v(q)$ is an increasing function of $q$.
\end{lemma}

Using Lemma \ref{lemma:inc}, we are ready to obtain an upper bound on the size of the set $S$ in terms of the probability $\Pr[S(h) \geq 1]$ that at least one configuration from $S$ survives after adding the randomly generated constraints. 
\begin{lemma}
\label{upperb:lemma}
Let $S \subseteq \{0,1\}^n$ and $h \in_R \mathcal{H}^f_{m \times n}$.
 Then 
\[
|S| \leq \min \left\{z \bigm \vert \frac{1}{1 + 2^{2m} v(z) / z^2}  > \Pr[S(h) \geq 1]\right\}
\]
\end{lemma}

The probability $\Pr[S(h) \geq 1]$ is unknown, but can be estimated from samples. In particular, we can draw independent samples of the hash functions and get accurate estimates using Chernoff style bounds. This yields the following theorem:

\begin{theorem}
\label{thm:UB}
Let $S \subseteq \{0,1\}^n$. Let $\Delta \in (0,1)$. Suppose we draw $T \geq 24 \ln \frac{1}{\Delta}$ 
hash functions $h_1, \cdots, h_T$ randomly from $\mathcal{H}^f_{m \times n}$. 
Let
\[
U(n,m,f) = \min \left\{z \bigm \vert \frac{1}{1 + 2^{2m} v(z) / z^2} \geq \frac{3}{4}\right\}
\]
Let $\mathcal{A}(S,h_1, \cdots, h_T)$ be a random variable that equals $U(n,m,f)$ if
$
\mathrm{Median}(\I[S(h_1)=0], \cdots, \I[S(h_T)=0])=1
$, and $2^n$ otherwise. Then
\begin{equation}
\label{eqn:ub}
\Pr\left[|S| \leq \mathcal{A}(S,h_1, \cdots, h_T)\right] \geq 1-\Delta
\end{equation}
\end{theorem}

Note that the upper bound on $|S|$ given by Theorem~\ref{thm:UB} is vacuous unless $\mathrm{Median}(\I[S(h_1)=0], \cdots, \I[S(h_T)=0])=1$. For brevity, let $\mathcal{E}$ denote this random event. $\mathcal{E}$ can be extremely unlikely for certain values of $m$, $|S|$, and $f$. For instance, if $|S|>0$ and $m=0$, event $\mathcal{E}$ is impossible, and the upper bound is necessarily vacuous. 
However, for $|S|>0$, as $m$ grows (i.e., more and more independent parity constraints are added), $\Pr[S(h)=0]$ increases and eventually approaches 1, making $\mathcal{E}$ increasingly likely. 
It is possible to show via Markov's inequality that, regardless of the value of $f$,  $\mathcal{E}$ is likely to occur when $m \geq \log|S|+c$ for some constant $c>0$. Thus, for large enough $m$, one expects Theorem~\ref{thm:UB} to provide a non-trivial upper bound on $|S|$.

When $f=0.5$, $v(z)=z / 2^m (1-1/2^m)$. This means that $U(n,m,0.5) \approx 3 \cdot 2^m$. 
For $m^* \approx \log|S|+c$ and $f=0.5$, Theorem \ref{thm:UB} provides a good upper bound on $|S|$ (with multiplicative error bounded by a constant factor). This is essentially the standard analysis for counting using hash functions in the pairwise independent case~\cite{trevisan2002lecture}.

When $f<0.5$, however,
it is possible that $U(n,m,f) \gg 2^m$, resulting in potentially loose upper bounds on $|S|$ even for small values of $m$.

We also note that when $f=0.5$, event $\mathcal{E}$ is unlikely to occur when $m \leq \log|S|-c$ for some constant $c$. When $f<0.5$, this need not be the case.

\section{Experimental Evaluation}

\subsection{Model Counting Benchmarks \footnote{Source code for this experiment can be found at https://github.com/ShengjiaZhao/XORModelCount} }

We evaluate the quality of our new bounds on a standard model counting benchmark (ANOR2011) from \namecite{kroc2011leveraging}. Both lower and upper bounds presented in the previous section are parametric: they depend both on $m$, the number of constraints, and $f$, the constraint density. Increasing $f$ is always beneficial, but can substantially increase the runtime. The dependence on $m$ is more complex, and we explore it empirically. To evaluate our new bounds, we consider a range of values for $f \in [0.01,0.5]$, and use a heuristic approach to first identify a promising value for $m$ using a small number of samples $T$, and then collect more samples for that $m$ to reliably estimate $P[S(h)\geq 1]$ and improve the bounds on $|S|$. 

We primarily compare our results to ApproxMC \cite{chakraborty2013scalable}, which can compute a constant factor approximation to user specified precision with arbitrary confidence (at the cost of more computation time). ApproxMC is similar to SPARSE-COUNT, and uses long parity constraints. For both methods, Cryptominisat version 2.9.4~\cite{soos2009extending} is used as the NP-oracle $\mathcal{O}_S$ (with Gaussian elimination enabled), and the confidence parameter is set to $0.95$, so that bounds reported hold with $95\%$ probability.

Our results on the \emph{Langford12} instance (576 variables, 13584 clauses, $10^5 \approx \exp(11.5)$ solutions) are shown in Figure~\ref{fig:SATLB}. The pattern is representative of all other instances we tested. The tradeoff between quality of the bounds and runtime, which is governed by $f$, clearly emerges. Instances with small $f$ values can be solved orders of magnitude faster than with full length XORs ($f \approx 0.5$), but provide looser bounds. Interestingly, lower bounds are not very sensitive to $f$, and we empirically obtain good bounds even for very small values of $f$. 
We also evaluate ApproxMC (horizontal and vertical lines) with parameter setting of $\epsilon=0.75$ and $\delta=0.05$, obtaining an 8-approximation with probability at least $0.95$. The runtime is $47042$ seconds. It can be seen that ApproxMC and our bounds offer comparable model counts for dense $f \approx 0.5$. 
However, our method allows to trade off computation time against the quality of the bounds. 
We obtain non-trivial upper bounds using as little as $0.1\%$ of the computational resources required with long parity constraints, a flexibility not offered by any other method.

Table~\ref{table:SATCount} summarizes our results on other instances from the benchmark and compares them with ApproxMC with a 12 hour timeout. We see that for the instances on which ApproxMC is successful, our method obtains approximate model counts of comparable quality and is generally faster. While ApproxMC requires searching for tens or even hundreds of solutions during each iteration, our method needs only one solution per iteration. Further, we see that long parity constraints can lead to very difficult instances that cannot be solved, thereby reinforcing the benefit of provable upper and lower bounds using sparse constraints (small $f$). Our method is designed to produce non-trivial bounds even when the computational budget is significantly limited, with bounds degrading gracefully with runtime. 

\FloatBarrier
\begin{figure}[ht]
    \centering
    \includegraphics[trim={5ex 5ex 5ex 14ex},clip,width=.95\columnwidth]{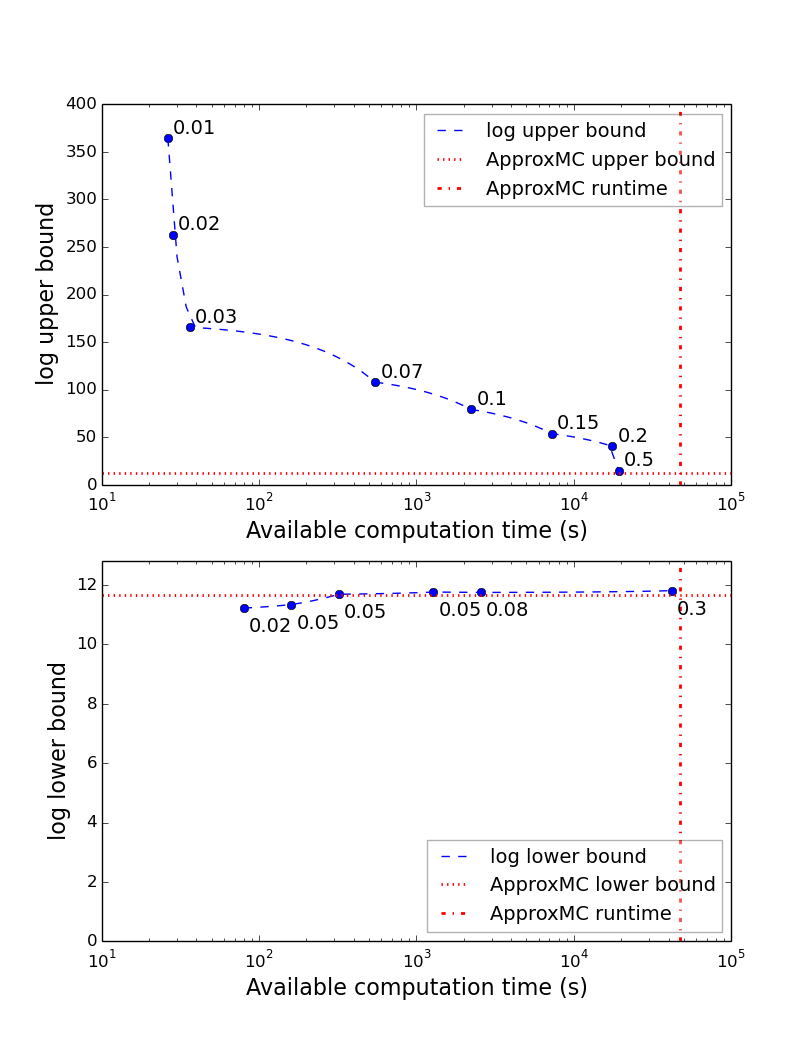}
    \caption{Log upper and lower bounds vs.\ computation time. Point labels indicate the value of $f$ used.}
    \label{fig:SATLB}
\end{figure}

\begin{table*}[t]
\caption{Comparison of run time and bound quality on the ANOR2011 dataset. All true counts and bounds are in log scale.}
\label{table:SATCount}
\centering
\begin{small}
\begin{tabular}{|lcc|ccc|ccc|cc|}
\hline
 & Num. & True & \multicolumn{3}{c|}{Upper Bound} \Tstrut & \multicolumn{3}{c|}{Lower Bound} & \multicolumn{2}{c|}{ApproxMC} \\
SAT Instance & Vars. & Count & UB  & Runtime(s) & $f$ & LB & Runtime (s) & $f$ & Estimate & Runtime (s) \\
\hline
\Tstrut
lang12         & 576  & -- & 14.96 & 19000 & 0.5 & 11.17 & 1280 & 0.1 & 12.25 & 47042 \\
wff.3.150.525  & 150  & 32.57 & 35.06 & 3800  & 0.5 & 31.69 & 3600 & 0.1  & 32.58 & 43571 \\
\hline
\Tstrut
2bitmax-6      & 252  & 67.52 & 69.72 & 355   & 0.5 & 67.17 & 1000 & 0.5  & -- & timeout \\
 lang15         & 1024 & -- & 352.5 & 240   & 0.02 & 19.58 & 6400 & 0.02 & -- & timeout \\
ls8-normalized & 301  & 27.01 & 201.5 & 3800 & 0.02 & 25.34 & 3600 & 0.02 & -- & timeout \\
wff.3.100.150  & 100  & 48.94 & 51.00 & 2100  & 0.5 & 48.30 & 3600 & 0.5  & -- & timeout \\
wff.4.100.500  & 100  & --    & 69.31 & 1100  & 0.05 & 37.90 & 3600 & 0.05 & -- & timeout \\
\hline
\end{tabular}
\end{small}
\end{table*}

\subsection{SMT Models for Contingency Table Analysis}

In statistics, a \emph{contingency table} is a matrix that captures the (multivariate) frequency distribution of two variables with $r$ and $c$ possible values, resp. For example, if the variables are gender (male or female) and handedness (left- or right-handed), then the corresponding $2 \times 2$ contingency table contains frequencies for the four possible value combinations, and is associated with $r$ row sums and $c$ column sums, also known as the row and column marginals.

Fisher's exact test \cite{fisher1954} tests contingency tables for homogeneity of proportion. Given fixed row and column marginals, it computes the probability of observing the entries found in the table under the null hypothesis that the distribution across rows is independent of the distribution across columns. This exact test, however, quickly becomes intractable as $r$ and $c$ grow. Statisticians, therefore, often resort to approximate significance tests, such as the chi-squared test.

The associated counting question is: \emph{Given $r$ row marginals $R_i$ and $c$ column marginals $C_j$, how many $r \times c$ integer matrices $M$ are there with these row and column marginals?} When entries of $M$ are restricted to be in $\{0,1\}$, the corresponding contingency table is called binary. We are interested in counting both binary and integer tables.

This counting problem for integer tables is \#P-complete even when $r$ is 2 or $c$ is 2~\cite{dyer1997sampling}. Further, for binary contingency tables with so-called structural zeros~\cite{chen2007conditional} (i.e., certain entries of $M$ are required to be $0$), we observe that the counting problem is still \#P-complete. This can be shown via a reduction from the well-known \emph{permanent} computation problem, which is \#P-complete even for 0/1 matrices~\cite{valiant1979permanent}. 

Model counting for contingency tables is formulated most naturally using integer variables and arithmetic constraints for capturing the row and column sums. While integer linear programming (ILP) appears to be a natural fit, ILP solvers do not scale very well on this problem as they are designed to solve optimization problems and not feasibility queries. We therefore propose to encode the problem using a Satisfiability Modulo Theory (SMT) framework~\cite{BarST-SMTLIB}, which extends propositional logic with other underlying theories, such as bitvectors and real arithmetic. We choose a bitvector encoding where each entry $a_{ij}$ of $M$ is represented as a bitvector of size $\lceil \log_2 \min(R_i, C_{j}) \rceil$. The parity constraints are then randomly generated over the individual bits of each bitvector, and natively encoded into the model as XORs. As a solver, we use Z3~\cite{de2008z3}.

We evaluate our bounds on six datasets:

\textbf{Darwin's Finches (df).}
The marginals for this binary contingency table dataset are from \citet{chen2005sequential}. This is one of the few datasets with known ground truth: $\log_2 |S| \approx 55.8982$, found using a clever divide-and-conquer algorithm of David desJardins. 
The 0-1 label in cell $(x, y)$ indicates the presence or absence of one of 13 finch bird species $x$ at one of 17 locations $y$ in the Gal\'apagos Islands. To avoid trivialities, we drop one of the species that appears in every island, resulting in $12 \times 17 = 204$ binary variables.

\textbf{Climate Change Perceptions (icons).}
This $9 \times 6$ non-binary dataset is taken from the \texttt{alymer} \texttt{R} package \cite{west2008exact}. It concerns lay perception of climate change. The dataset is based on a study reported by \citet{oneil2008:thesis} in which human subjects are asked to identify which icons (such as polar bears) they find the most concerning. There are 18 structural zeros representing that not all icons were shown to all subjects.

\textbf{Social Anthropology (purum).}
This $5 \times 5$ non-binary dataset~\cite{west2008exact} concerns marriage rules of an isolated tribe in India called the Purums, which is subdivided into $5$ sibs. Structured zeros represent marriage rules that prevent some sibs from marrying other sibs. 

\textbf{Industrial Quality Control (iqd).}
This $4 \times 7$ non-binary dataset~\cite{west2008exact} captures an industrial quality control setting. Cell $(x,y)$ is the number of defects in the $x$-th run attributable to machine $y$. It has 9 structured zeros, representing machines switched off for certain runs.

\textbf{Synthetic Data (synth).}
This $n \times n$ binary dataset contains \emph{blocked matrices}~\cite{golshan2013row}. The row and column marginals are both $\{1,n-1, \ldots, n-1\}$. It can be seen that a blocked matrix has either has a value of 1 in entry $(1, 1)$ or it has two distinct entries with value 1 in the first row and the first column, cell $(1, 1)$ excluded. Instantiating the first row and the first column completely determines the rest of the table. It is also easy to verify that the desired count is $1+(n-1)^2$.

\begin{table}[t]
\caption{Lower (LB) and upper (UB) bounds on $\log_2 |S|$.
The trivial upper bound (Trv.~UB) is the number of binary variables. $f^\ast$ denotes the best previously known minimum $f$ \cite{ermon2014low} required for provable upper bounds.}
\label{summary-table}
\centering
\begin{small}
\addtolength{\tabcolsep}{-0.7ex}
\begin{tabular}{lcc|rrrr}
\hline
\Tstrut
Dataset & Table Size & $f^\ast$ & LB ($f$) & $\log_2 |S|$ & UB & Trv.~UB\\
\hline
df  \Tstrut  & 12 $\times$ 17  & 0.18 & {\bf 53} (0.03) & 55.90 & {\bf 150} & 204 \\
icons & 9 $\times$ 6    & 0.19 & {\bf 58} (0.04) & - & {\bf 183} & 236 \\
purum & 5 $\times$ 5    & 0.26 & {\bf 29} (0.13) & - & {\bf 52}  & 125 \\
iqd   & 4 $\times$ 7    & 0.34 & {\bf 15} (0.10) & - & {\bf 17}  & 76  \\
synth\_8 & 8 $\times$ 8 & 0.41 & {\bf 5}  (0.30) & 5.64 & {\bf 16} & 64\\
synth\_20 & 20 $\times$ 20 & 0.42 & {\bf 8} (0.40) & 8.49 & {\bf 14} & 400\\
\hline
\end{tabular}
\end{small}
\end{table}

Table \ref{summary-table} summarizes the obtained lower and upper bounds on the number of contingency tables, with a 10 minute timeout. For the datasets with ground truth, we see that very sparse parity constraints (e.g., $f=0.03$ for the Darwin finches dataset, as opposed to a theoretical minimum of $f^*=0.18$) often suffice in practice to obtain very accurate lower bounds. For the iqd dataset, we obtain upper and lower bounds within a small constant factor. For other datasets, there is a wider gap between the upper and lower bounds. However, the upper bounds we obtain are orders of magnitude tighter than the trivial log-upper bounds, which is the number of variables in a binary encoding of the problem.

\section{Conclusions}

We introduced a novel analysis of the randomized hashing schemes used by numerous recent approximate model counters and probabilistic inference algorithms. We close a theoretical gap, providing a tight asymptotic estimate for the minimal constraint density required. Our analysis also shows, for the first time, that even very short parity constraints can be used to generate non-trivial upper bounds on model counts. Thanks to this finding, we proposed a new scheme for computing anytime upper and lower bounds on the model count. Asymptotically, these bounds are guaranteed to become tight (up to a constant factor) as the constraint density grows. Empirically, given very limited computational resources, we are able to obtain new upper bounds on a variety of benchmarks, including a novel application for the analysis of statistical contingency tables.

A promising direction for future research is the analysis of related ensembles of random parity constraints, such as low-density parity check codes~\cite{achlioptasstochastic}. 

\section{Acknowledgments}
This work was supported by the Future of Life Institute (grant 2015-143902).

\bibliographystyle{aaai}
\begin{small}
\bibliography{xorcount}

\begin{thebibliography}{}

\bibitem[\protect\citeauthoryear{Achlioptas and
  Jiang}{2015}]{achlioptasstochastic}
Achlioptas, D., and Jiang, P.
\newblock 2015.
\newblock Stochastic integration via error-correcting codes.
\newblock In {\em Proc. Uncertainty in Artificial Intelligence}.

\bibitem[\protect\citeauthoryear{Angluin and Valiant}{1979}]{angluin1979fast}
Angluin, D., and Valiant, L.
\newblock 1979.
\newblock Fast probabilistic algorithms for hamiltonian circuits and matchings.
\newblock {\em Journal of Computer and System Sciences} 18(2):155--193.

\bibitem[\protect\citeauthoryear{Barrett, Stump, and
  Tinelli}{2010}]{BarST-SMTLIB}
Barrett, C.; Stump, A.; and Tinelli, C.
\newblock 2010.
\newblock {The Satisfiability Modulo Theories Library (SMT-LIB)}.
\newblock {\tt www.SMT-LIB.org}.

\bibitem[\protect\citeauthoryear{Belle, Van~den Broeck, and
  Passerini}{2015}]{belle2015hashing}
Belle, V.; Van~den Broeck, G.; and Passerini, A.
\newblock 2015.
\newblock Hashing-based approximate probabilistic inference in hybrid domains.
\newblock In {\em Proceedings of the 31st Conference on Uncertainty in
  Artificial Intelligence (UAI)}.

\bibitem[\protect\citeauthoryear{Biere \bgroup et al\mbox.\egroup
  }{2009}]{biere2009handbook}
Biere, A.; Heule, M.; van Maaren, H.; and Walsh, T.
\newblock 2009.
\newblock Handbook of satisfiability. frontiers in artificial intelligence and
  applications, vol. 185.

\bibitem[\protect\citeauthoryear{Calabro}{2009}]{calabro2009entropy}
Calabro, C.
\newblock 2009.
\newblock {\em The Exponential Complexity of Satisfiability Problems}.
\newblock Ph.D. Dissertation, University of California, San Diego.

\bibitem[\protect\citeauthoryear{Chakraborty, Meel, and Vardi}{2013a}]{cav13}
Chakraborty, S.; Meel, K.; and Vardi, M.
\newblock 2013a.
\newblock A scalable and nearly uniform generator of {SAT} witnesses.
\newblock In {\em Proc. of the 25th International Conference on Computer Aided
  Verification (CAV)}.

\bibitem[\protect\citeauthoryear{Chakraborty, Meel, and
  Vardi}{2013b}]{chakraborty2013scalable}
Chakraborty, S.; Meel, K.; and Vardi, M.
\newblock 2013b.
\newblock A scalable approximate model counter.
\newblock In {\em Proc. of the 19th International Conference on Principles and
  Practice of Constraint Programming (CP)},  200--216.

\bibitem[\protect\citeauthoryear{Chen \bgroup et al\mbox.\egroup
  }{2005}]{chen2005sequential}
Chen, Y.; Diaconis, P.; Holmes, S.~P.; and Liu, J.~S.
\newblock 2005.
\newblock Sequential monte carlo methods for statistical analysis of tables.
\newblock {\em Journal of the American Statistical Association}
  100(469):109--120.

\bibitem[\protect\citeauthoryear{Chen}{2007}]{chen2007conditional}
Chen, Y.
\newblock 2007.
\newblock Conditional inference on tables with structural zeros.
\newblock {\em Journal of Computational and Graphical Statistics} 16(2).

\bibitem[\protect\citeauthoryear{Cooper}{2000}]{cooper2000rank}
Cooper, C.
\newblock 2000.
\newblock On the rank of random matrices.
\newblock {\em Random Structures \& Algorithms} 16(2):209--232.

\bibitem[\protect\citeauthoryear{De~Moura and Bj{\o}rner}{2008}]{de2008z3}
De~Moura, L., and Bj{\o}rner, N.
\newblock 2008.
\newblock Z3: An efficient smt solver.
\newblock In {\em Tools and Algorithms for the Construction and Analysis of
  Systems}. Springer.
\newblock  337--340.

\bibitem[\protect\citeauthoryear{Dyer, Kannan, and
  Mount}{1997}]{dyer1997sampling}
Dyer, M.; Kannan, R.; and Mount, J.
\newblock 1997.
\newblock Sampling contingency tables.
\newblock {\em Random Structures and Algorithms} 10(4):487--506.

\bibitem[\protect\citeauthoryear{Ermon \bgroup et al\mbox.\egroup
  }{2013a}]{uai13LPCount}
Ermon, S.; Gomes, C.~P.; Sabharwal, A.; and Selman, B.
\newblock 2013a.
\newblock Optimization with parity constraints: From binary codes to discrete
  integration.
\newblock In {\em Proc. of the 29th Conference on Uncertainty in Artificial
  Intelligence (UAI)}.

\bibitem[\protect\citeauthoryear{Ermon \bgroup et al\mbox.\egroup
  }{2013b}]{wishicml13}
Ermon, S.; Gomes, C.~P.; Sabharwal, A.; and Selman, B.
\newblock 2013b.
\newblock Taming the curse of dimensionality: Discrete integration by hashing
  and optimization.
\newblock In {\em Proc. of the 30th International Conference on Machine
  Learning (ICML)}.

\bibitem[\protect\citeauthoryear{Ermon \bgroup et al\mbox.\egroup
  }{2014}]{ermon2014low}
Ermon, S.; Gomes, C.~P.; Sabharwal, A.; and Selman, B.
\newblock 2014.
\newblock Low-density parity constraints for hashing-based discrete
  integration.
\newblock In {\em Proc. of the 31st International Conference on Machine
  Learning (ICML)},  271--279.

\bibitem[\protect\citeauthoryear{Fisher}{1954}]{fisher1954}
Fisher, R.
\newblock 1954.
\newblock {\em Statistical Methods for Research Workers}.
\newblock Oliver and Boyd.

\bibitem[\protect\citeauthoryear{Gogate and
  Dechter}{2007}]{gogate2007approximate}
Gogate, V., and Dechter, R.
\newblock 2007.
\newblock {Approximate counting by sampling the backtrack-free search space}.
\newblock In {\em Proc. of the 22nd National Conference on Artifical
  Intelligence (AAAI)}, volume~22,  198--203.

\bibitem[\protect\citeauthoryear{Gogate and
  Domingos}{2011}]{gogate2011probabilistic}
Gogate, V., and Domingos, P.
\newblock 2011.
\newblock Probabilistic theorem proving.
\newblock In {\em Uncertainty in Artificial Intelligence}.

\bibitem[\protect\citeauthoryear{Golshan, Byers, and
  Terzi}{2013}]{golshan2013row}
Golshan, B.; Byers, J.; and Terzi, E.
\newblock 2013.
\newblock What do row and column marginals reveal about your dataset?
\newblock In {\em Advances in Neural Information Processing Systems},
  2166--2174.

\bibitem[\protect\citeauthoryear{Gomes \bgroup et al\mbox.\egroup
  }{2007}]{ghss07:shortxors}
Gomes, C.~P.; Hoffmann, J.; Sabharwal, A.; and Selman, B.
\newblock 2007.
\newblock Short {XORs} for model counting: From theory to practice.
\newblock In {\em Theory and Applications of Satisfiability Testing (SAT)},
  100--106.

\bibitem[\protect\citeauthoryear{Gomes, Sabharwal, and Selman}{2006}]{mbound}
Gomes, C.~P.; Sabharwal, A.; and Selman, B.
\newblock 2006.
\newblock Model counting: A new strategy for obtaining good bounds.
\newblock In {\em Proc. of the 21st National Conference on Artificial
  Intelligence (AAAI)},  54--61.

\bibitem[\protect\citeauthoryear{Guruswami}{2010}]{codlect}
Guruswami, V.
\newblock 2010.
\newblock Introduction to coding theory - lecture notes.

\bibitem[\protect\citeauthoryear{Ivrii \bgroup et al\mbox.\egroup
  }{2015}]{ivrii2015computing}
Ivrii, A.; Malik, S.; Meel, K.~S.; and Vardi, M.~Y.
\newblock 2015.
\newblock On computing minimal independent support and its applications to
  sampling and counting.
\newblock {\em Constraints}  1--18.

\bibitem[\protect\citeauthoryear{Kolchin}{1999}]{kolchin1999random}
Kolchin, V.~F.
\newblock 1999.
\newblock {\em Random graphs}.
\newblock Number~53 in Encyclopedia of Mathematics and its Applications.
  Cambridge University Press.

\bibitem[\protect\citeauthoryear{Kroc, Sabharwal, and
  Selman}{2011}]{kroc2011leveraging}
Kroc, L.; Sabharwal, A.; and Selman, B.
\newblock 2011.
\newblock Leveraging belief propagation, backtrack search, and statistics for
  model counting.
\newblock {\em Annals of Operations Research} 184(1):209--231.

\bibitem[\protect\citeauthoryear{O'Neil}{2008}]{oneil2008:thesis}
O'Neil, S.
\newblock 2008.
\newblock {\em An Iconic Approach to Communicating Climate Change}.
\newblock Ph.D. Dissertation, School of Environmental Science, University of
  East Anglia.

\bibitem[\protect\citeauthoryear{Richardson and
  Domingos}{2006}]{richardson2006markov}
Richardson, M., and Domingos, P.
\newblock 2006.
\newblock Markov logic networks.
\newblock {\em Machine Learning} 62(1):107--136.

\bibitem[\protect\citeauthoryear{Sang, Beame, and
  Kautz}{2005}]{sang2005solving}
Sang, T.; Beame, P.; and Kautz, H.
\newblock 2005.
\newblock Solving {B}ayesian networks by weighted model counting.
\newblock In {\em Proc. of the 20th National Conference on Artificial
  Intelligence (AAAI)}, volume~1,  475--482.

\bibitem[\protect\citeauthoryear{Sang \bgroup et al\mbox.\egroup
  }{2004}]{sang2004combining}
Sang, T.; Bacchus, F.; Beame, P.; Kautz, H.; and Pitassi, T.
\newblock 2004.
\newblock {Combining component caching and clause learning for effective model
  counting}.
\newblock In {\em Theory and Applications of Satisfiability Testing (SAT)}.

\bibitem[\protect\citeauthoryear{Sheldon and
  Dietterich}{2011}]{sheldon2011collective}
Sheldon, D.~R., and Dietterich, T.~G.
\newblock 2011.
\newblock Collective graphical models.
\newblock In {\em Advances in Neural Information Processing Systems},
  1161--1169.

\bibitem[\protect\citeauthoryear{Sinclair}{2011}]{rclect}
Sinclair, A.
\newblock 2011.
\newblock Randomness and computation - lecture notes.

\bibitem[\protect\citeauthoryear{Soos, Nohl, and
  Castelluccia}{2009}]{soos2009extending}
Soos, M.; Nohl, K.; and Castelluccia, C.
\newblock 2009.
\newblock Extending {SAT} solvers to cryptographic problems.
\newblock In {\em Theory and Applications of Satisfiability Testing (SAT)}.

\bibitem[\protect\citeauthoryear{Stockmeyer}{1985}]{stockmeyer1985approximation}
Stockmeyer, L.
\newblock 1985.
\newblock On approximation algorithms for \#{P}.
\newblock {\em SIAM Journal on Computing} 14(4):849--861.

\bibitem[\protect\citeauthoryear{Trevisan}{2004}]{trevisan2002lecture}
Trevisan, L.
\newblock 2004.
\newblock Lecture notes on computational complexity.

\bibitem[\protect\citeauthoryear{Tseitin}{1968}]{tseitin68}
Tseitin, G.~S.
\newblock 1968.
\newblock On the complexity of derivation in the propositional calculus.
\newblock In Slisenko, A.~O., ed., {\em Studies in Constructive Mathematics and
  Mathematical Logic, Part II}.

\bibitem[\protect\citeauthoryear{Valiant}{1979a}]{valiant1979permanent}
Valiant, L.~G.
\newblock 1979a.
\newblock The complexity of computing the permanent.
\newblock {\em Theoretical computer science} 8(2):189--201.

\bibitem[\protect\citeauthoryear{Valiant}{1979b}]{valiant1979complexity}
Valiant, L.
\newblock 1979b.
\newblock The complexity of enumeration and reliability problems.
\newblock {\em SIAM Journal on Computing} 8(3):410--421.

\bibitem[\protect\citeauthoryear{West and Hankin}{2008}]{west2008exact}
West, L.~J., and Hankin, R.~K.
\newblock 2008.
\newblock Exact tests for two-way contingency tables with structural zeros.
\newblock {\em Journal of Statistical Software} 28(11):1--19.

\bibitem[\protect\citeauthoryear{Zhao \bgroup et al\mbox.\egroup
  }{2015}]{zhaotech}
Zhao, S.; Chaturapruek, S.; Sabharwal, A.; and Ermon, S.
\newblock 2015.
\newblock Closing the gap between short and long xors for model counting.
\newblock Technical report, Stanford University.

\end{thebibliography}
\end{small} 

\section{Appendix: Proofs}

\begin{theorem}[Formal statement of Theorem~\ref{thm:const}] 
Let $\{\mathcal{H}^{f_i}_{i \times n}\}_{i=0}^n$ be families of $f_i$-sparse hash functions.
\begin{itemize}

\item (Sufficiency) If there exist $c \geq 2$ and $\epsilon > 0$ such that for all $i$, $\mathcal{H}^{f_i}_{i \times n}$ $\epsilon$-shatters all sets $S' \subseteq \{0,1\}^n$ of size $|S'|=2^{i + c}$, then for any set $S$, $0<\Delta<1$, and $\alpha \leq 2\, \left(\min(\epsilon,1/2-1/2^c)\right)^2 \ln 2$, SPARSE-COUNT$(\mathcal{O}_S, \Delta, \alpha, \{\mathcal{H}^{f_i}_{i \times n}\} )$ outputs a $2^{c+1}$ approximation of $|S|$
with probability at least $1-\Delta$.

\item (Necessity) If there exists an $i$ and a set $S$ of size $2^{i+c}$ such that for all $\epsilon>0$ $\mathcal{H}^{f_i}_{i \times n}$ does not $\epsilon$-shatter $S$, then for any choice of $\alpha>0$ and $0<\Delta<1$, SPARSE-COUNT$(\mathcal{O}_S, \Delta, \alpha, \{\mathcal{H}^{f_i}_{i \times n}\} )$ outputs a $2^c$ approximation of $|S|$
with probability at most $1/2$.
\end{itemize}
\end{theorem}
\begin{proof}

For the sufficiency part, we show that there exist $c > 0$ and $\delta > 2$ such that for all $i$ two conditions hold:\\
(a) for all sets $S \subseteq \{0,1\}^n$ of size $|S| \leq 2^{i - c}$
\[ \Pr[S(h) = 0] \geq 1 - \frac{1}{\delta} \]
 when $h$ is chosen from $\mathcal{H}^{f_i}_{i \times n}$\\
(b)  for all sets $S \subseteq \{0,1\}^n$ of size $|S| \geq 2^{i + c}$
\[ \Pr[S(h) \geq 1] \geq 1 - \frac{1}{\delta} \]
when $h$ is chosen from $\mathcal{H}^{f_i}_{i \times n}$. 
Standard analysis following \namecite{ermon2014low} then implies that for any set $S$ and $0<\Delta<1$, if $\alpha \leq 2(1-\frac{1}{\delta} - \frac{1}{2})^2
\ln 2$, we have
\[
\frac{|S|}{2^{c+1}} \leq \text{SPARSE-COUNT}(\mathcal{O}_S, \Delta, \alpha, \{\mathcal{H}^{f_i}_{i \times n}\} ) \leq |S|2^{c}
\]
with probability at least $1-\Delta$. 

The second condition (b) is implied by the shattering properties of $h$ in the assumptions for some $c=c'$ and $\epsilon=1/2-\frac{1}{\delta}$.

The first condition (a) is trivially satisfied for any $c \geq 2$ and $\delta=2^c$. 
Formally, we have
\[ \begin{aligned}
	\Pr[S(h) > 0] &= \Pr[S(h) \geq 1] = \Pr[S(h) \geq 2^c \mu_S]  \\
    & \leq  \frac{1}{2^{c}}
    \end{aligned}\]
from Markov's inequality. Conditions (a) and (b) are therefore simultaneously met choosing $c=c'$ and $\delta=\min(2^c,\frac{1}{1/2-\epsilon})$. 

For the necessity part, let $S$ be a set of size $2^{i+c}$ as in the statement of the Theorem, i.e., not shattered by $\mathcal{H}^{f_i}_{i \times n}$. Let us condition on the event that the outer loop of SPARSE-COUNT$(\mathcal{O}_S, \Delta, \alpha, \{\mathcal{H}^{f_i}_{i \times n}\} )$ reaches iteration $i$. For any $T \geq 1$ (therefore, for any choice of $\Delta$ and $\alpha$), the while loop breaks at iteration $i$ with probability at least $1/2$ because by assumption $\Pr[S(h)\geq 1] \leq 1/2$ when $h$ is chosen from $\mathcal{H}^{f_i}_{i \times n}$. Otherwise, $\mathcal{H}^{f_i}_{i \times n}$ would $\epsilon$-shatter $S$ for some $\epsilon>0$. Therefore, the output satisfies $\frac{|S|}{2^{c}} \leq$ SPARSE-COUNT$(\mathcal{O}_S, \Delta, \alpha, \{\mathcal{H}^{f_i}_{i \times n}\})$ with probability at most $1/2$. This also bounds the probability that the output is a $2^c$ approximation of $|S|$.
\end{proof}


\section{Proofs of New Upper and Lower Bounds}

We continue with proofs for the bounds on $|S|$ for arbitrary constraint density $f$. Several of the following proofs will rely on the following notion:
\begin{mydef}[\citenobrackets{ermon2014low}]
\label{def:wstar-etc}
Let $m, n \in \mathbb{N}, m \leq n, f \leq \half,$ and $q \leq 2^n + 1$. Then:
\begin{align}
w^*(n,q) & = \max\left\{w \mid \sum_{j = 1}^{w} \binom{n}{j} \leq q - 1\right\}\\
r(n,q) & = \left(q - 1 - \sum_{w = 1}^{w^*(n,q)} \binom{n}{w}\right)\\
\epsilon(n, m, q, f) & = \frac{1}{q - 1} \left[ \sum_{w = 1}^{w^*(n,q)} \nonumber \binom{n}{w} \frac{1}{2^m} (1 + (1 - 2f)^w)^m \right. \\
    & \ \ \ \ \ \ \ \left. +\, \frac{r}{2^m} (1 + (1 - 2f)^{w^*(n,q) + 1})^m \right] 
\end{align}
\end{mydef}
We observe that $w^*(n,q)$ is always at most $n$. We will often be interested in the case where $q = 2^{m+c}$ for $c \geq 0$.

\subsection{New Lower Bound}

\begin{proof}[Proof of Lemma~\ref{lemma:lb_ideal}]
Let hash function $h_m$ be drawn randomly from $\mathcal{H}^f_{m \times n}$. Recall that our random variable $S(h_m)$ takes on a non-negative integer value. Then, for any $m$, by Markov's inequality,
\begin{align*}
\Pr[S(h_m) \geq 1] \leq \Exp[S(h_m)] = \frac{|S|}{2^m}.
\end{align*}
Hence, for any $m$, $|S| \geq 2^m \Pr[S(h_m) \geq 1]$. Taking the maximum over all choices of $m$ finishes the proof.
\end{proof}

\begin{proof}[Proof of Theorem~\ref{thm:tighter-LB}]


We observe that $\mathbb{E}\left[Y\right]=T\mathbb E\left[Y_{1}\right]=
T\Pr\left[S(h_{m}^1)\geq1\right].$ 
There are two cases: either $\mathbb E\left[Y\right]>cT /(1+\kappa)$ or $\Exp[Y] \leq cT /(1+\kappa)$.

Case 1: $\Exp[Y] > cT/(1+\kappa)$. This implies $\Pr[S(h_{m}^1) \geq 1] > c/(1+\kappa)$. From Lemma~\ref{lemma:lb_ideal}, $|S| \geq 2^m c/(1+\kappa)$
which is no smaller than $\mathcal{B}$ for all realizations of the random variables $h_{m}^1, \cdots, h_{m}^T$. Thus, in this case,
$
\Pr\left[|S|\geq \mathcal{B}\right] =1
$.

Case 2: $\Exp[Y] \leq cT /(1+\kappa)$. Define $Z = \sum_{k=1}^T Z_k$ where $Z_k$ are i.i.d.\ Bernoulli variables with probability $c/(1+\kappa)$ of being 1. Then $\Exp[Z] = cT/(1+\kappa)$. Since $Y_k$ and $Z_k$ are i.i.d.\ Bernoulli random variables with $\Exp[Y_k] \leq \Exp[Z_k]$, we have $\Pr[Y \geq d] \leq \Pr[Z \geq d]$ for any $d$. Thus:
\begin{align*}
\Pr\left[\Pr_{\text{est}}\left[S(h)\geq1\right] \geq c\right] & = \Pr[Y \geq c T] \\
  & \leq \Pr[Z \geq c T] \\
  & = \Pr[Z \geq (1+\kappa)\, \Exp[Z]] \\
  & \leq \exp\left(-\frac{\kappa^2}{2+\kappa} \frac{cT}{1+\kappa} \right)
\end{align*}
where the last inequality follows from Chernoff's bound~\cite{angluin1979fast,rclect}.
Hence, with a probability of at least $1 - \exp\left(-\frac{\kappa^2 cT}{(1+\kappa)(2+\kappa)} \right)$, 
$\mathcal{B}$ must be $0$. We therefore have:
\[
\Pr\left[|S|\geq\mathcal{B}\right] \geq 1 - \exp\left(-\frac{\kappa^2 cT}{(1+\kappa)(2+\kappa)} \right)
\]
This finishes the proof.
\end{proof}

\subsection{New Upper Bound}

\begin{proof}[Proof of Lemma~\ref{lemma:inc}]
Let $f(q)=q^2/v(q)$, where 
\[
v(q)=\frac{q}{2^{m}}\left(1+\epsilon(n, m, q, f)\cdot(q-1)-\frac{q}{2^{m}}\right)
\]
is defined as in Definition \ref{def:vq}. We show that $v(q)<v(q+1)$ for all $q$. Removing constant terms
in $v(q)$ we see that it suffices to show that 
\begin{equation}
g(q)=\frac{q}{B_1 + B_2 - q},
\end{equation}
where $B_1 = 2^{m}+\sum_{w=1}^{w^*(n,q)}{n \choose w}(1+x^{w})^{m}$ and $B_2 = \left(q-1-\sum_{w=1}^{w^*(n,q)}{n \choose w}\right)(1+x^{w^*(n,q)+1})^{m}$,
is an increasing function of $q$. The relevant quantities are defined in Definition \ref{def:wstar-etc}, and $x=1-2f$ for brevity. We show that $g(q)<g(q+1).$ We
note that
\begin{equation}
w^*(n,q+1)=w^*(n,q)+h(q),
\end{equation}
where $h(q)\in\left\{ 0,1\right\} $ and it is $1$ only when $\sum_{j=1}^{w^*(n,q)+1}{n \choose j}=q$
(by looking at the definition). Define 
\begin{align}
&\nonumber t(q)=B_1 + B_2 - q=2^{m}+\sum_{w=1}^{w^*(n,q)}{n \choose w}(1+x^{w})^{m}+\\
&\qquad\left(q-1-\sum_{w=1}^{w^*(n,q)}{n \choose w}\right)(1+x^{w^*(n,q)+1})^{m}-q.
\end{align}

\begin{description}
\item [{Case~1:~$h(q)=0$.}] We have $g(q)<g(q+1)$ if and only if
\begin{equation}
\frac{q}{t(q)}<\frac{q+1}{t(q)+(1+x^{w^*(n,q)+1})^{m}-1},
\end{equation}
which is true if and only if
\begin{equation}
(1+x^{w^*(n,q)+1})^{m}q-q<t(q).
\end{equation}
Expanding the definition of $t(q),$ we get that the above inequality
is true if and only if
\begin{align}
&\nonumber 0<\left(2^{m}-(1+x^{w^*(n,q)+1})^{m}\right)+\\
&\qquad\sum_{w=1}^{w^*(n,q)}{n \choose w}\left((1+x^{w})^{m}-(1+x^{w^*(n,q)+1})^{m}\right),
\end{align}
which is true because both terms are positive.
\item [{Case~2:~$h(q)=1$.}] This implies $\sum_{j=1}^{w^*(n,q)+1}{n \choose j}=q.$
We have $g(q)<g(q+1)$ if and only if
\begin{equation}
\frac{q}{t(q)}<\frac{q+1}{t(q+1)},
\end{equation}
and we have 
\begin{equation}
t(q+1) = 2^{m}+\sum_{w=1}^{w^*(n,q)+1}{n \choose w}(1+x^{w})^{m}-q-1.
\end{equation}
Expanding the definition of $t(q)$ and $t(q+1),$ we get that the
above inequality is true if and only if
\begin{align}
&\nonumber 0<\left(2^{m}-(1+x^{w^*(n,q)+1})^{m}\right)+\\
&\qquad\sum_{w=1}^{w^*(n,q)}{n \choose w}\left((1+x^{w})^{m}-(1+x^{w^*(n,q)+1})^{m}\right),
\end{align}
which is true because both terms are positive (note this is the same inequality as before).
\end{description}
\end{proof}

\begin{proof}[Proof of Lemma~\ref{upperb:lemma}]
Let $Q\subseteq\{0,1\}^n$ be any set of size exactly $q$ and $h \in_R \mathcal{H}^f_{m \times n}$. Following \namecite{ermon2014low}, we can get a worst-case bound for the variance of $Q(h)$ as a function of $q$. Regardless of the structure of $Q$, we have
\[
\sigma^2(Q) \leq v(q)= \frac{q}{2^m} \left(1 + \epsilon(n,m,q,f)(q-1) -\frac{q}{2^m}\right) 
\]
where $\sigma^2(Q)= \mathrm{Var}[|\{x\in Q \mid h(x) =0\}|]$ is the variance of the random variable $Q(h)$, and $\epsilon(n,m,q,f)$ is from Definition~\ref{def:wstar-etc}. 
From Cantelli's inequality
\[
\Pr[Q(h)>0] \geq 1- \frac{\sigma^2(Q)}{\sigma^2(Q)+\left(\frac{|Q|}{2^m}\right)^2} \geq 1- \frac{v(q)}{v(q)+\left(\frac{q}{2^m}\right)^2}
\]
We claim that
\[
\frac{v(q)}{v(q)+\left(\frac{q}{2^m}\right)^2}
\]
which gives a lower bound on the shattering probability, is a decreasing function of $q$. By dividing numerator and denominator by $v(q)$, it is sufficient to show that
\[
\frac{\left(\frac{q}{2^m}\right)^2}{v(q)}
\]
is increasing in $q$, which follows from Lemma \ref{lemma:inc}. To prove the Lemma, suppose by contradiction that
\[
|S| > \min \left\{x \bigm \vert 1 - \frac{v(x)}{v(x)+\left(\frac{x}{2^m}\right)^2}  > \Pr\left[S(h)>0\right]\right\}
\]
Since $\frac{v(q)}{v(q)+\left(\frac{q}{2^m}\right)^2}$ is a decreasing function of $q$, and $|S|$ is assumed to be larger than the smallest element in the set above, it holds that
\begin{equation}
\label{eq:cont}
1 - \frac{v(|S|)}{v(|S|)+\left(\frac{|S|}{2^m}\right)^2} > \Pr\left[S(h)>0\right]
\end{equation}
From Cantelli's inequality
\[
\Pr[S(h)>0] \geq 1- \frac{\sigma^2(S)}{\sigma^2(S)+\left(\frac{|S|}{2^m}\right)^2} \geq 1- \frac{v(|S|)}{v(|S|)+\left(\frac{|S|}{2^m}\right)^2}
\]
where the second inequality holds because $v(|S|)$ upper bounds the true variance $\sigma^2(S)$. The last inequality contradicts eq.~(\ref{eq:cont}).
\end{proof}

\begin{proof}[Proof of Theorem~\ref{thm:UB}]
There are two possibilities for the unknown value $|S|$: either $|S|\leq U(n,m,f)$ or $|S| > U(n,m,f)$. Note that, by definition, $|S| \leq 2^n$.

Case 1: $|S|\leq U(n,m,f)$. In this case, $|S| \leq \min\{U(n,m,f), 2^n\} \leq \mathcal{A}(S,h_1, \cdots, h_T)$ for all realizations of the random variables $h_1, \cdots, h_T$. Thus, eq.~(\ref{eqn:ub}) holds trivially for any $\Delta \geq 0$.

Case 2: $|S| > U(n,m,f)$. In this case, rearranging terms, we obtain
\[ 
|S| > \min \left\{z \bigm \vert 1 - \frac{v(z)}{v(z)+\left(\frac{z}{2^m}\right)^2}  \geq \frac{3}{4}\right\}
\]
By the monotonicity of $\frac{v(q)}{v(q)+\left(\frac{q}{2^m}\right)^2}$ shown in the proof of Lemma ~\ref{upperb:lemma}, it holds that
\[
1 - \frac{v(|S|)}{v(|S|)+\left(\frac{|S|}{2^m}\right)^2} \geq \frac{3}{4}
\]
From Cantelli's inequality 
\begin{eqnarray*}
\Pr[S(h)>0] &\geq& 1- \frac{\mathrm{Var}(S(h))}{\mathrm{Var}(S(h))+\left(\frac{|S|}{2^m}\right)^2} \\
&\geq& 1- \frac{v(|S|)}{v(|S|)+\left(\frac{|S|}{2^m}\right)^2}\\
&\geq& \frac{3}{4}
\end{eqnarray*}
where the second inequality is derived observing that $v(|S|)$ is an upper bound on the true variance $\mathrm{Var}(S(h))$. 
Thus $\Pr[S(h)=0] \leq \frac{1}{4}$, 
and from Chernoff's bound
\begin{eqnarray*}
\Pr[\mathrm{Median}(\I[S(h_1)=0], \cdots, \I[S(h_T)=0])=1] \leq\\
\exp\left(-T/24\right) \leq \Delta
\end{eqnarray*}
By the definition of $\mathcal{A}(S,h_1, \cdots, h_T)$, we therefore have
\[
\Pr\left[\mathcal{A}(S,h_1, \cdots, h_T) = 2^n\right] \geq 1 - \Delta
\]
Since $|S| \leq 2^n$, eq.~(\ref{eqn:ub}) follows immediately.
\end{proof}

\section{Proof of Theorem~\ref{thm:asymp}}

Theorem~\ref{thm:asymp} contains three statements:
\begin{enumerate}
\item Let $\epsilon > 0$. $\kappa > 1$. There exists $M_\kappa$ such that for all $m \geq M_\kappa$:
\[ \minf_\epsilon(m,n) > \frac{\log{m}}{\kappa\, m}\]

\item Let $\epsilon \in (0,\frac{3}{10}), \alpha \in (0,1),$ and $m = \alpha n$. There exists $N$ such that for all $n \geq N$:
\[ \minf_\epsilon(m,n) \leq \left(\kbound\right) \frac{\log{m}}{m} \]

\item Let $\epsilon \in (0,\frac{3}{10}), \alpha, \kappa > 1, \beta \in (0,1),$ and $m = \alpha n^\beta$. There exists $N_\kappa$ such that for all $n \geq N_\kappa$:
\[ \minf_\epsilon(m,n) \leq \frac{\kappa\, (1-\beta)}{2\beta}\frac{\log^2{m}}{m} \]
\end{enumerate}

We will prove them in turn in the following subsections. The arguments will often use the inequality $1+x \leq \exp(x)$ which holds for any $x \in \mathbb{R}$. We also use the expression "for sufficiently large n" to mean the more formal statement 
\[ \exists N > 0, \forall n > N \]

\subsection{Part I: Lower Bound}

\begin{proof}[Proof of Theorem~\ref{thm:asymp} (part 1)]
Since the minimum constraint density must work for every set $S$, it must also work for the hypercube $S_c = \{0,1\}^{m+c} \times \{0\}^{n-m-c}$ with $2^{m+c}$ elements. This is a set where the first $m+c$ variables are ``free'', while the remaining $n-m-c$ are fixed to $0$.

Let $h$ be a hash function drawn from $\mathcal{H}_{m \times n}^f$. Let's consider a parity constraint of the form $a_{i1}x_1 \oplus \cdots \oplus a_{in}x_n = b_i$ as in Definition \ref{def:sparse-hash}. If $a_{i1} = a_{i2} = ... = a_{i(m+c)} = 0$ and $b_i = 1$, then $\{x \in S_c : a x =b \bmod 2\} = \emptyset$. If the constraint is constructed as in Definition \ref{def:sparse-hash}, this happens with probability $\half (1 - f)^{m+c}$. Accumulating this probability over $m$ independent parity constraints and setting $f = \frac{\log m}{\kappa\,m}$ for any $\kappa > 1$, we obtain:
\begin{align*}
\Pr[S_c(h) > 0]
  & \leq \left(1 - \half (1 - f)^{m+c}\right)^m \\
  & = \left(1 - \half \left(1 - \frac{\log m}{\kappa m}\right)^{m+c}\right)^m
\end{align*}
%
For any $\lambda > 1$, it can be verified from the Taylor expansion of the exponential function that for any small enough $x > 0$, $1 - x \geq \exp(-\lambda x)$.
Observe that for any fixed $\kappa>1$, $1 - \frac{\log m}{\kappa\,m} > 0$ as long as $m$ is large enough. It follows that for any $\gamma > 1$, there exists an $M_{\kappa,\gamma}$ such that for all $m \geq M_{\kappa,\gamma}$, the above expression is upper bounded by:
\begin{align*}
  & \left(1 - \half \exp\left(-\gamma \frac{\log m}{\kappa\,m} (m+c)\right)\right)^m \\
  & = \left(1 - \half m^{-\gamma (m+c)/(\kappa\,m)}\right)^m \\
  & \leq \exp\left(- \half m^{1-\gamma (m+c)/(\kappa m)} \right)
\end{align*}
where the last inequality follows from $1 + x \leq \exp(x)$. Since $\kappa > 1$, we can choose $\gamma$ such that $1<\gamma < \kappa$. In this case, for large enough $m$, the last expression above is less than $1/2$. In other words, there exists an $M_\kappa$ such that for all $m \geq M_\kappa$, $\Pr[S_c(h) > 0] < 1/2$. It follows that for all such $m$, the minimum constraint density, $\minf(m,n)$, must be larger than $\frac{\log m}{\kappa\,m}$, finishing the proof.
\end{proof}
\subsection{Part II: Upper bound when $m = \Theta(n)$}

To prove the upper bound for $m = \Theta(n)$,  we will need to first establish a few lemmas. In all the proofs below we will assume $m = \alpha n$, with $\alpha$ constant with respect to $n$. Let us denote the binary entropy function as $H(p) \triangleq -p \log_2 p -(1-p) \log_2 (1-p)$. It is well known that $H(0)=0$, $H(\half)=1$, and it is monotonically increasing in the interval $[0,\half]$. We use the following relationship between the sum of binomials and the binary entropy function:
\begin{proposition}[\citeauthor{codlect} \citeyear{codlect}, Lemma 5]
\label{prop:binomial-sum-entropy}
For any $n \in \mathbb{N}$ and $\lambda \in [0,\half]$,
\[ \sum_{j=0}^{\lambda n} {n \choose j} \leq 2^{H(\lambda)\,n}  \]
\end{proposition}

\begin{lemma}
\label{lemma:binom_bound}
Let $\alpha \in (0,1)$, there exists a unique $\lambda^* < \half$ such that $H(\lambda^*) = \alpha$, where $H$ is the binary entropy function. For all $\lambda < \lambda^*$,
\[ \limn \frac{\sum_{j=1}^{\lambda n} \binom{n}{j}}{2^{\alpha n}} = 0 \]
\end{lemma}

\begin{proof}
We can always find a unique $\lambda^* < \half$ such that $H(\lambda^*)=\alpha$. This is because $H(\lambda)$ increases monotonically from 0 to 1 as $\lambda$ increases from 0 to $\half$, so $H^{-1}(\alpha)$ takes one and only one value in the range $(0, 1/2)$. 

Furthermore, due to monotonicity, $H(\lambda) < \alpha=H(\lambda^*)$ for all $\lambda < \lambda^*$. From Proposition~\ref{prop:binomial-sum-entropy}, for any $\lambda < \half$, the sum of binomials in the numerator of the desired quantity is at most $2^{H(\lambda)\, n}$. Hence, the fraction is at most $2^{(H(\lambda)-\alpha)\, n}$, which approaches $0$ as $n$ increases because $H(\lambda) < \alpha$. Since numerator and denominator are non-negative, the limit is zero and this concludes the proof.
\end{proof}

\begin{corollary}
\label{cor:wasymp}
Let $\alpha \in (0,1), c \geq 2, w^*(n,q)$ be as in Definition~\ref{def:wstar-etc},
and $\lambda^* < \half$ be such that $H(\lambda^*) = \alpha$. Then for all $\lambda < \lambda^*$, and any $n$ sufficiently large
\[ w^*(n,2^{m+c}) = w^*(n,2^{\alpha n+c}) \geq \lambda n \]
\end{corollary}

\begin{proof}[Proof of Corollary~\ref{cor:wasymp}]
By Lemma~\ref{lemma:binom_bound}, for all $\lambda < H^{-1}(\alpha)$, when n is sufficiently large, 
\[ \sum_{j = 1}^{\lambda n}\binom{n}{j} < 2^{\alpha n} < 2^{m+c} - 1\]
Thus, it follows immediately from the definition of $w^*$ that for sufficiently large $n$, $\lambda n \leq w^*(n, 2^{m+c})$.
\end{proof}

\begin{remark}
Corollary~\ref{cor:wasymp}, together with the trivial fact that $w^*(n,q) \leq n$, implies $w^*(n,q) = \Theta(n)$ when $m=\alpha n$ and $q=2^{m+c}$.
\end{remark}

\begin{lemma}
\label{lemma:convex}
For all $\delta > 0$ and $w \in \mathbb{R}$, the function $f_\delta(w) = \log{(1 + \delta^w)} $ is convex. \par
\end{lemma}
\begin{proof}
We will show that the second derivative of $f_\delta(w)$ is non-negative:
\[ \begin{aligned} 
	f_\delta'(w) &= \frac{\delta^w \log{\delta}}{1 + \delta^w} \\
    f_\delta''(w) &= \frac{\delta^w(1+\delta^w)(\log{\delta})^2 - \delta^{2w}(\log{\delta})^2}{(1+\delta^w)^2} \\
    &= \frac{\delta^w(\log{\delta})^2}{(1+\delta^w)^2} \geq 0
    \end{aligned}
\]
It follows that $f_\delta(w)$ is convex.
\end{proof}

\begin{lemma}
\label{lemma:ubound}
Let $t > 0, 0<\delta<1, k > -t\frac{\log{(\frac{2}{1+\delta} - 1)}}{\log{(1+\delta)}},$ and $w \geq 0$. Then for all $m$ sufficiently large,
\begin{equation}
\label{finallimitnn}
\frac{\left(\half + \half \left(1 - \frac{k \log{m}}{m} \right)^w \right)^m}{m^{-tw} + (1+\delta)^{-m}} < 1
\end{equation}
\end{lemma}

Lemma~\ref{lemma:ubound} is an attempt to simplify the expression below, which we will call $\zeta(w)$ and make its dependence on $k$ and $m$ implicit.
\begin{align*}
\zeta(w) = \zeta(w, k, m) = \left(\half + \half \left(1 - \frac{k \log{m}}{m} \right)^w \right)^m
\end{align*}
Note that for large enough $m$ such that $\frac{k \log m}{m} \leq 1$, $\zeta(w)$ is monotonically non-increasing in $w$, a property we will use.
This term is too complex to study in detail, therefore, we upper bound it by the sum of two simpler expressions. The intuition for this bound is shown in Figure~\ref{fig:lemma_ubound}.

\begin{figure}[h]
    \centering
    \includegraphics[width=\columnwidth]{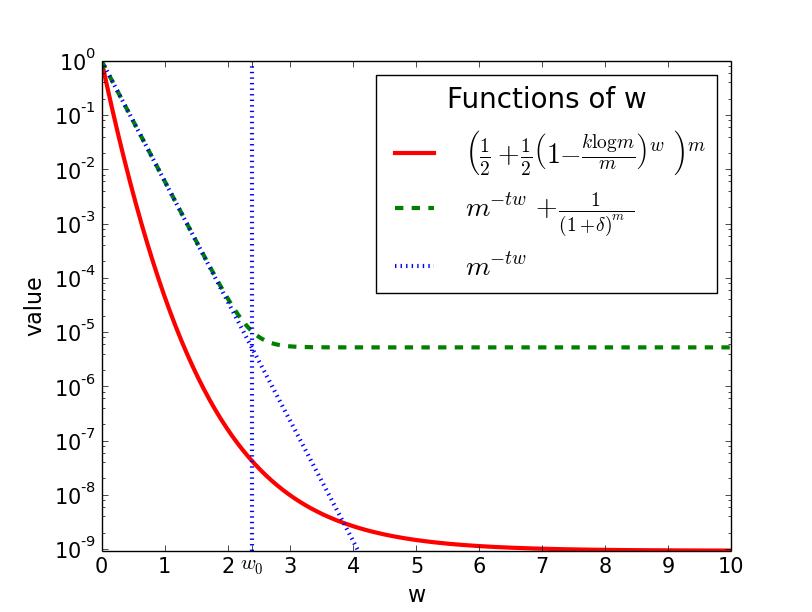}
    \caption{Graphical representation of Lemma~\ref{lemma:ubound}}
    \label{fig:lemma_ubound}
\end{figure}

Lemma~\ref{lemma:ubound} can thus be restated as claiming $\zeta(w) \leq m^{-tw} + (1+\delta)^{-m}$ for $m$ sufficiently large.
Towards this end, when $w < w_0$, we show that $m^{-tw}$ is the dominant term and that $\zeta(w) < m^{-tw}$. When $w > w_0$, we show that the term $(1+\delta)^{-m}$ dominates and that $\zeta(w) < (1+\delta)^{-m}$. Combining these two regimes, we deduce that $\zeta(w)$ must be upper bounded by their sum for all values of $w$. A formal proof follows.

\begin{proof}
We will show that for $m$ sufficiently large, $\zeta(w) \leq m^{-tw} + (1+\delta)^{-m}$. Assume w.l.o.g.\ that $m$ is large enough to satisfy:
\begin{align*}
1 - k\frac{\log{m}}{m} \geq 0 \numberthis \label{equ:pos_cond} \\
\end{align*}

Next we consider the location $w_0$ where the dominant term of $m^{-tw} + (1+\delta)^{-m}$ switches from $m^{-tw}$ to $(1+\delta)^{-m}$. This is where 
\[ m^{-tw_0} = (1+\delta)^{-m} \]
which gives us $w_0 = \frac{m\log{(1 + \delta)}}{t\log{m}}$. At this $w_0$ we have
\[ \begin{aligned}
	& \frac{\zeta(w_0)}{(1+\delta)^{-m}} \\
    & \leq \frac{(1 + \delta)^{m}}{2^m} \left( 1 + \left(1 - \frac{k\log{m}}{m}\right)^{w_0} \right)^m \\
    & \leq \frac{(1 + \delta)^{m}}{2^m} \left(1 + \exp\left(-\frac{k \log m}{m} \, \frac{m\log{(1 + \delta)}}{t\log{m}}\right) \right)^m \\
    &= \frac{(1 + \delta)^{m}}{2^m} \left(1 + \exp\left(-\frac{k}{t}\log{(1 + \delta)}\right)\right)^m \\
    &= \frac{(1 + \delta)^{m}}{2^m} \left(1 + (1 + \delta)^{-\frac{k}{t}}\right)^m \\
    &= \left( \frac{(1 + \delta)(1 + (1 + \delta)^{-\frac{k}{t}})}{2} \right)^m \\
    \end{aligned}     
    \]
Clearly if we choose $k$ such that $\frac{(1 + \delta)(1 + (1 + \delta)^{-\frac{k}{t}})}{2} < 1$, then the entire expression is less then $1$. This condition is satisfied if: 
\[ k > -t\frac{\log{(\frac{2}{1+\delta} - 1)}}{\log{(1+\delta)}} \]
Recall our earlier observation that $\zeta(w)$ is monotonically non-increasing in $w$ for large enough $m$. Thus, for $m$ sufficiently large and any $w \geq w_0$, we have $\zeta(w) \leq \zeta(w_0) < (1+\delta)^{-m} \leq m^{-tw} + (1+\delta)^{-m}$.

Finally, let's consider the case where $w  w_0$, again assuming $m$ sufficiently large so that $\frac{k \log m}{m} < 1$. Notice that \[
\log{\zeta(w)} = m\, \log \left(\half + \half \left(1 - \frac{k \log{m}}{m} \right)^w \right)
\]
is convex with respect to $w$ because of Lemma \ref{lemma:convex}. 
We have that for all positive $m$:
\[ \log{\zeta(0)} = \log(1) = \log{ \left[ m^{-t0} \right] } \]
and for all $m$ sufficiently large:
\[ \log{\zeta(w_0)} < \log{\left[(1 + \delta)^{-m}\right]} = \log{\left[m^{-t w_0}\right]} \]
where the inequality is from the above analysis and the equality is by definition of $w_0$.

For $w \in [0,w_0]$, we can write $w=(1-\lambda)0+\lambda w_0$ for some $\lambda \geq 0$. Therefore, for such $w$ and for $m$ sufficiently large, by convexity of $\log \zeta(w)$:
\[ \begin{aligned}
	\log\zeta(w) &\leq (1-\lambda)\log\zeta(0) + \lambda\log\zeta(w_0) \\
    &\leq (1-\lambda)\log[m^{-t0}] + \lambda \log{[m^{-tw_0}]} \\
    &= \log{\left[ m^{-tw} \right]} 
	\end{aligned} \]
i.e., $\zeta(w) \leq m^{-tw} < m^{-tw} + (1+\delta)^{-m}$, as desired.
\end{proof}

\begin{lemma}
\label{lemma:partial}
Let $\alpha, \delta \in (0,1), k > -\frac{\log{\left(\frac{2}{1+\delta} - 1\right)}}{\log{\left(1+\delta\right)}}$, and $\lambda^* < \half$ be such that $H(\lambda^*) = \alpha \log_2{(1+\delta)}$. Then, for all $\lambda < \lambda^*$,
\[ \limn \sum_{w = 1}^{\lambda n} \binom{n}{w} \frac{1}{2^m} \left(1 + \left(1 - 2\frac{k\log{m}}{m}\right)^w\right)^m = 0 \] 
\end{lemma}

\begin{proof}
By Lemma~\ref{lemma:ubound}, we can select any $0 < \delta < 1$, $t > 1$ and $k > -t\frac{\log{(\frac{2}{1+\delta} - 1)}}{\log{(1+\delta)}}$ so that when n (or equivalently $m = \alpha n$) is sufficiently large, 
\[ \begin{aligned}
    \sum_{w = 1}^{\lambda n} & \binom{n}{w} \frac{1}{2^m} \left(1 + \left(1 - 2\frac{k\log{m}}{m}\right)^w\right)^m \\ 
    & \leq \sum_{w = 1}^{\lambda n} \binom{n}{w} \left( m^{-tw} + (1+\delta)^{-m} \right) \\
    & \leq \sum_{w = 1}^{\lambda n} \frac{n^w}{w!}{(\alpha n)}^{-tw} + \frac{\sum_{w = 1}^{\lambda n} \binom{n}{w}}{(1 + \delta)^m} 
    \end{aligned} \]
    where we used the inequality $\binom{n}{w} \leq \frac{n^w}{w!}$ for all $n \in \mathbb{N}^*$ and $0 \leq w \leq n$.
The first term of the sum can be driven to zero because when we choose any $t > 1$
\[ \begin{aligned}
	\limn \sum_{w = 1}^{\lambda n} \frac{n^w}{w!}{(\alpha n)}^{-tw} 
    &= \limn \sum_{w = 1}^{\lambda n} \frac{{\alpha}^{-tw}n^{(1-t)w}}{w!} \\
    &\leq \limn \sum_{w = 0}^{\infty} \frac{{\alpha}^{-tw}n^{(1-t)w}}{w!} - \frac{{\alpha}^0n^0}{0!}  \\
	&= \limn e^{\alpha^{-t}n^{1-t}} - 1 = 0
	\end{aligned} \]
This requires that there exists $t > 1$ such that $k > -t\frac{\log{(\frac{2}{1+\delta} - 1)}}{\log{(1+\delta)}}$. For this, 
$k > -\frac{\log{(\frac{2}{1+\delta} - 1)}}{\log{(1+\delta)}} $ suffices.

The second term can be driven to zero when $H(\lambda) < \alpha \log_2{(1+\delta)}$ as a direct consequence of Lemma~\ref{lemma:binom_bound}.
\end{proof}

\begin{lemma}
\label{lemma:asymp}
Let $\alpha \in (0,1), m = \alpha n$, $c \in \mathbb{N}, q = 2^{m + c}$. Let $w^*$ and $\epsilon(n, m, q, f)$ be as in Definition~\ref{def:wstar-etc}. Then, for all $\gamma > 1$, $k \geq \kbound$ and $f = \frac{k\log{m}}{m}$, $\exists N_k > 0$, so that $\forall n \geq N_k$, we have
\[ \epsilon(n, m, q, f) \leq \gamma \frac{2^c}{q-1} \] 
\end{lemma}

\begin{proof}
For any $\delta \in (0,1)$, if we choose $\lambda^* < \half$, such that $H(\lambda^*) = \alpha \log_2{(1+\delta)}$, then $\forall \lambda < \lambda^*$
by Corollary~\ref{cor:wasymp}, we have for any value of $n$ sufficiently large,
$ \lambda n \leq w^*(n,q)$.
Thus: 
\[ 
\begin{aligned} 
    &(q-1)\,\epsilon(n, m, q, f) \\
    &= \sum_{w = 1}^{\lambda n} \binom{n}{w} \frac{1}{2^m} (1 + (1 - 2f)^w)^m + \\
    & \qquad \sum_{w = \lambda n + 1}^{w^*} \binom{n}{w} \frac{1}{2^m} (1 + (1 - 2f)^w)^m + \\
    & \qquad \frac{r}{2^m} (1 + (1 - 2f)^{w^* + 1})^m \\
    &\leq \sum_{w = 1}^{\lambda n} \binom{n}{w} \frac{1}{2^m} (1 + (1 - 2f)^w)^m + \\
    & \qquad \sum_{w = 1}^{w^*} \binom{n}{w} \frac{1}{2^m} (1 + (1 - 2f)^{\lambda n})^m + \\
    & \qquad \frac{r}{2^m} (1 + (1 - 2f)^{\lambda n})^m \\
    &= \sum_{w = 1}^{\lambda n} \binom{n}{w} \frac{1}{2^m} (1 + (1 - 2f)^w)^m + \\
    & \qquad \frac{q-1}{2^m} (1 + (1 - 2f)^{\lambda n})^m \\
    &= A_n + B_n 
\end{aligned} \]
where $A_n = \sum_{w = 1}^{\lambda n} \binom{n}{w} \frac{1}{2^m} (1 + (1 - 2f)^w)^m$ and $B_n = \frac{q-1}{2^m} (1 + (1 - 2f)^{\lambda n})^m$.
By our choice of $\lambda$, and according to Lemma~\ref{lemma:partial}, if we choose any $k > -\frac{\log{(\frac{2}{1+\delta} - 1)}}{\log{(1+\delta)}}$ and $f = \frac{k\log{m}}{m}$, we can have
\[ \limn A_n = 0 \]
For $B_n$ we have
\[\begin{aligned}
    B_n &= \frac{q-1}{2^m} \left(1 + \left(1 - 2\frac{k\log{m}}{m}\right)^{\lambda n}\right)^m \\
    &\leq \frac{q-1}{2^m} \left(1 + \exp\left(-\frac{2k\log{m}}{m}\lambda n\right)\right)^m \\
    &\leq 2^c \left(1 + m^{-\frac{2k\lambda}{\alpha}}\right) ^ m \\
    &\leq 2^c \exp\left(m^{1 - \frac{2k\lambda}{\alpha}}\right)
\end{aligned} \]
where the inequalities follow from $1+x \leq \exp(x)$.

If we choose $k$ such that $1 - \frac{2k\lambda}{\alpha} < 0$, or equivalently $k > \frac{\alpha}{2\lambda}$, we have
\[ \lim\sup_{n \to \infty} B_n \leq 2^c \]

If we choose a $k$ that is sufficiently large to satisfy both $k > -\frac{\log{(\frac{2}{1+\delta} - 1)}}{\log{(1+\delta)}}$ and $k > \frac{\alpha}{2\lambda}$, we have
\[ \lim\sup_{n \to \infty} A_n + B_n \leq 2^c \]
which implies that for all $\gamma > 1$, and $n$ sufficiently large,
\begin{align*}
\epsilon(n, m, q, f) &\leq \frac{A_n + B_n}{q-1} \leq \gamma \frac{2^c}{q-1} 
\end{align*}

Now we obtain an upper bound on the value of $k$ (the constraint density $f$ is proportional to $k$, so we'd like this number to be as small as possible). From the derivation above, we can choose any $0 < \delta < 1$, and any $k$ that satisfy the following inequalities:
\begin{equation} 
\label{firstineq} k > -\frac{\log{(\frac{2}{1+\delta} - 1)}}{\log{(1+\delta)}} \end{equation}
\begin{equation} 
\label{seconineq}
k > \frac{\alpha}{2\lambda} \end{equation}

The second inequality also depends on $\lambda$, which we are free to choose as long as it satisfies $\lambda < \lambda^*$, or
\[ H(\lambda) < \log_2{(1+\delta)} \alpha \equiv \sigma \]
We denote the latter term as $\sigma$ to lighten the notation. This is satisfied if 
\[ \lambda < \frac{\sigma}{2\log_2(6/\sigma)} \leq H^{-1}(\sigma) \]
The latter inequality is adapted from Theorem 2.2 in \cite{calabro2009entropy}

Therefore, considering that $\alpha<1$, the following condition is tighter than (\ref{seconineq})
\[ k > \frac{\log_2(6/\sigma)}{\log{(1+\delta)}} \]
Combining with (\ref{firstineq}) we have the following condition on $k$:
\[ k > \max \left( -\frac{\log{(\frac{2}{1+\delta} - 1)}}{\log{(1+\delta)}}, \quad \frac{\log_2(6/\sigma)}{\log_2(1+\delta)} \right) \]
We are allowed to choose $\delta$ to give us the best bound. This choice is asymptotically insignificant, so we choose an arbitrary but empirically well performing $\delta=3/4$, and derive 
\[ k > \max \left( -\frac{\log{(\frac{2}{7/4} - 1)}}{\log{(7/4)}}, \quad \frac{\log_2{\frac{6}{\log_2{7/4}}} - \log_2{\alpha}}{\log_2(7/4)} \right) \]
which is approximately
\[ k > \max \left( 3.47, \quad  3.58 - 1.23 \log_2{\alpha} \right) \]
which is implied by
\begin{equation} 
\label{kupper}
k \geq 3.6 - \frac{5}{4}\log_2{\alpha}
\end{equation}
as desired.
\end{proof}

The bounds in Lemma~\ref{lemma:asymp} are graphically shown in Figure~\ref{fig:asymp_linear}. When $m = \alpha n$, the plot on the left shows the minimum $f^*$ so that $\epsilon(n, m, q, f^*)$ defined in Definition~\ref{def:wstar-etc} is less than $2 / 2^m$. The plot on the right shows the empirical k so that the $f^* \equiv k\frac{\log{m}}{m}$. We also show the proved asymptotic bounds $k = 3.6 - \frac{5}{4}\log_2{\alpha}$ for comparison. As expected, the value of $k$ found empirically does not exceed the bound (\ref{kupper}). 

\begin{figure}[h]
    \centering
    \includegraphics[width=\columnwidth]{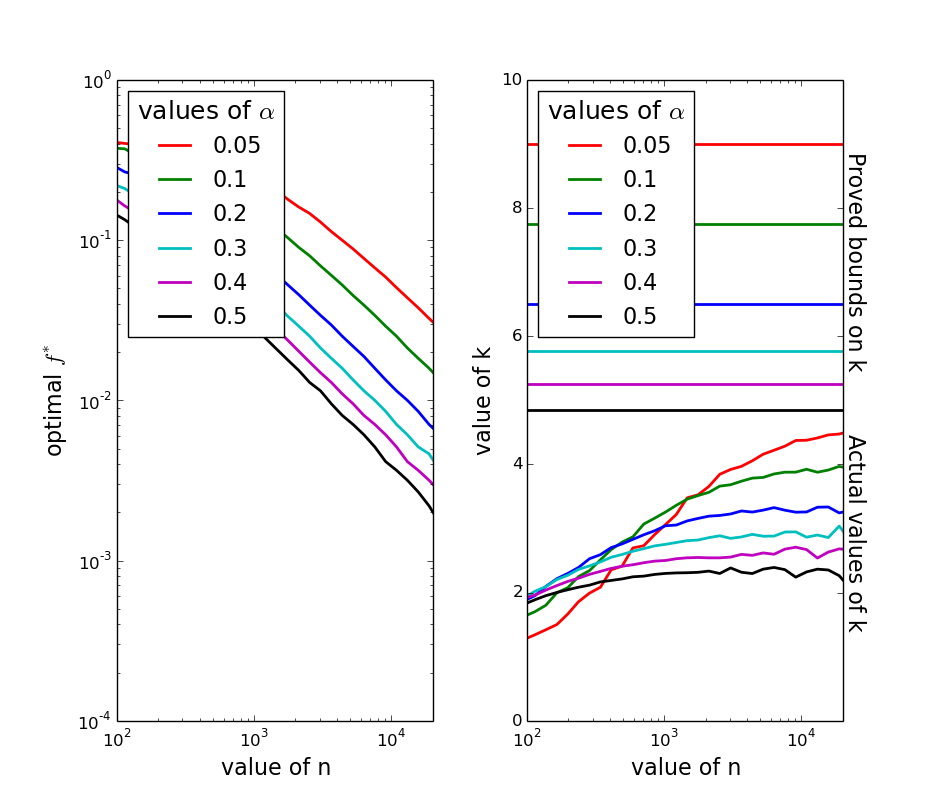}
    \caption{Empirical vs. proved bound on the value of k}
    \label{fig:asymp_linear}
\end{figure}

\begin{proof}[\textbf{Proof of Theorem~\ref{thm:asymp} (part 2)}]
By Corollary~1 and Theorem~2 of \namecite{ermon2014low},
for set $S$ with size $|S|=q=2^{m+c}$ and $h \in \mathcal{H}^f_{m \times n}$, a sufficient condition for ensuring that $S$ is $\varepsilon$-shattered, i.e., $\Pr[S(h) \geq 1] \geq \half + \varepsilon$, is the ``weak-concentration'' condition given by:\footnote{Note that the notation used for $1/2 + \varepsilon$ (for $\varepsilon > 0$) by \namecite{ermon2014low} is $1 - 1/\delta$ (for $\delta > 2$).}
\begin{align}
\epsilon(n, m, |S|, f)
  & \leq \frac{\mu / (\half + \varepsilon) - 1}{|S| - 1}  
    = \frac{2^c / (\half + \varepsilon) - 1}{q-1}
    \label{eqn:eps-requirement}
\end{align}
By Lemma~\ref{lemma:asymp}, when $\gamma > 1$, $f > (\kbound) \frac{\log{m}}{m}, c \geq 2$, and $m$ is sufficiently large:
\[ \epsilon(n, m, |S|, f) \leq \gamma \frac{2^c}{q-1} \]
Hence, to satisfy requirement~(\ref{eqn:eps-requirement}), it suffices to have:
\[ \gamma 2^c \leq \frac{2^c}{\half + \varepsilon} - 1 \]
that is, $\gamma \leq 1/(\half + \varepsilon) - 2^{-c}$. We can therefore choose a $\gamma > 1$ whenever $1/(\half + \varepsilon) > 1 + 2^{-c}$. Rearranging terms, this yields $\varepsilon < \frac{2^c - 1}{2(2^c + 1)}$. Hence, for $c \geq 2$, it suffices to have $\varepsilon < 3/10$.
This completes the proof of part two of Theorem~\ref{thm:asymp}.
\end{proof}

\subsection{Part III: Upper Bound when $m = \Theta(n^\beta)$}

Similar to Part II, we will first establish a few lemmas. We will assume for the rest of the reasoning that $m = \alpha n^\beta$ for some constant $\alpha, \beta \in (0,1)$. 

\begin{lemma}
\label{lemma:binom_bound2}
Let $\alpha, \beta \in (0,1), \gamma > 0, m = \alpha n^\beta$, and $\lambda^* = \frac{\gamma}{1-\beta}$. Then for all $\lambda < \lambda^*$,
\[ \limn \frac{\sum_{j=1}^{\lambda m/\log n} \binom{n}{j}}{2^{\gamma m}} = 0 \]
\end{lemma}

\begin{proof}
For any $1 \leq w \leq \frac{n}{2}$,
\begin{align*}
\log \left( \sum_{j=1}^w \binom{n}{j} \right)
  & \leq \log \left( w \binom{n}{w} \right)
    \leq \log \left( w \left( \frac{ne}{w} \right)^w \right) \\
  & \leq  \log w + w \log \frac{ne}{w}
    \leq w \log \frac{2ne}{w}
\end{align*}
When n is sufficiently large, $1 \leq \lambda \frac{m}{\log{n}} \leq \frac{n}{2}$. Let $\epsilon > 0$ be any constant. Substituting $w = \lambda m = \lambda \frac{\alpha n^\beta}{\log n}$:
\begin{align*}
\log \left( \sum_{j=1}^{\lambda \alpha n^\beta/\log n} \binom{n}{j} \right)
  & \leq \frac{\lambda \alpha n^\beta}{\log n} \log \left( \frac{2en^{1-\beta} \log n}{\lambda \alpha} \right) \\
  & \leq (1-\beta+\epsilon) \lambda \alpha n^\beta
\end{align*}
for large enough $n$. We thus have,
\begin{align*}
\limn \frac{\sum_{j=1}^{\lambda m/\log n} \binom{n}{j}}{2^{\gamma m}}
  & \leq \limn 2^{\left((1-\beta+\epsilon)\lambda - \gamma\right) \alpha n^\beta}
\end{align*}
Let $\lambda^* = \gamma / (1-\beta)$. It follows that for any $\lambda < \lambda^*$, $\exists \epsilon > 0$, such that
\begin{align*}
(1-\beta+\epsilon)\lambda < \frac{1 - \beta + \epsilon}{1 - \beta}\gamma < \gamma
\end{align*}
and the above limit can be driven to zero.
\end{proof}

\begin{corollary}
\label{cor:wasymp2}
Let $\alpha, \beta \in (0, 1), m = \alpha n^\beta, c \geq 2,$ $w^*$ as in Definition~\ref{def:wstar-etc}, and $\lambda^* = 1 / (1-\beta)$. Then, for all $\lambda < \lambda^*$, $\exists N > 0$ such that $\forall n > N$, we have
\[ \frac{\lambda m}{\log{n}} \leq w^*(n, 2^{m+c}) \] \par
\end{corollary}

\begin{lemma}
\label{lemma:ubound2}
Let $t > 0, \delta \in (0,1), w \geq 0$, and $k > 0$. Then for all values of $m$ sufficiently large,
\[ \frac{\left(\half + \half \left(1 - \frac{k \log^2{m}}{m} \right)^w \right)^m}{m^{-tw} + (1 + \delta)^{-m}} < 1  \]
\end{lemma}

\begin{proof}
Similar to Lemma~\ref{lemma:ubound}, we will simplify notation by defining
\begin{align*}
\zeta(w) = \left(\half + \half \left(1 - \frac{k \log^2{m}}{m} \right)^w \right)^m
\end{align*}
and assume that $1 - \frac{k\log^2{m}}{m} > 0$. Consider the bound at 
\begin{align*}
w_0 = \frac{m\log{(1 + \delta)}}{t\log{m}}
\end{align*}
where $m^{-tw} = (1 + \delta)^{-m}$, then separately consider the cases of $w$ being smaller and larger than $w_0$. When $w = w_0$, we have:
\begin{align*}
  & \frac{\zeta(w_0)}{(1+\delta)^{-m}} \\
  & \leq \frac{(1 + \delta)^m}{2^m} \left(1 + \exp \left(\frac{- w_0 k \log^2{m}}{m}\right) \right)^m \\
  & = \frac{(1 + \delta)^m}{2^m} \left(1 + \exp \left(\frac{-k\log{ (1+\delta)}\log{m}}{t}\right) \right)^m \\
  & = \frac{(1 + \delta)^m}{2^m} \left(1 + (1 + \delta)^{-\frac{k}{t}\log{m}} \right)^m \\
  & = \left( \frac{(1 + \delta)(1 + (1 + \delta)^{-\frac{k}{t}\log{m}})}{2} \right)^m
\end{align*}
It is easy to see that for any $k > 0, t > 0$, when $m$ is sufficiently large the base of this exponential quantity, and hence the quantity itself, is smaller than $1$.

Now consider the general case of $w > w_0$. Because $m > 0$ and $\zeta(w)$ is monotonically non-increasing in $w$, we have:
\begin{align*}
  \frac{\zeta(w)}{m^{-tw} + (1 + \delta)^{-m}} < \frac{\zeta(w_0)}{(1+\delta)^{-m}} 
\end{align*}
which, by the above argument, is smaller than $1$, as desired.

The remaining case of $w \leq w_0$ is proved similar to the proof of Lemma~\ref{lemma:ubound}. Due to the convexity of $\log{\zeta(w)}$ with respect to $w$, combined with the fact that for any $m > 0$:
\[ \log{\zeta(0)} = \log{ \left[ m^{-t0} \right] } \]
and for $m$ sufficiently large:
\[ \log{\zeta(w_0)} \leq \log{\left[ m^{-tw_0} \right]} \]
we have that $\exists M$ such that $\forall m>M, 0 \leq w \leq w_0$
\begin{align*}
\log\zeta(w) \leq \log{\left[ m^{-tw} \right]}
\end{align*}
which implies 
\[ \zeta(w) < m^{-tw} + (1+\delta)^{-m} \]
when $w \leq w_0$. Combined with the earlier similar result for $w > w_0$, this finishes the proof.
\end{proof}

\begin{lemma}
\label{lemma:partial2}
Let $\alpha, \beta, \delta \in (0,1), m = \alpha n^\beta, \lambda^* = \frac{\log_2(1+\delta)}{1-\beta}, \lambda < \lambda^*$, and $k > 0$. Then
\[ \limn \sum_{w = 1}^{\lambda m/\log n} \binom{n}{w} \frac{1}{2^m} \left(1 + \left(1 - 2\frac{k\log^2{m}}{m}\right)^w\right)^m = 0 \] 
\end{lemma}

\begin{proof}
By Lemma~\ref{lemma:ubound2}, for any $t > 0$ and large enough $n$ (and thus $m$), the desired expression is at most:
\begin{align*}
  & \sum_{w = 1}^{\lambda m/\log n} \binom{n}{w} \left( m^{-tw} + (1 + \delta)^{-m} \right) \\
  & \leq \sum_{w = 1}^{\lambda m/\log n} \frac{n^w m^{-tw}}{w!} + \frac{\sum_{w = 1}^{\lambda m/\log n} \binom{n}{w}}{(1 + \delta)^m}
\end{align*}

The second term here converges to zero as $n \to \infty$
by applying Lemma~\ref{lemma:binom_bound2} with $\gamma$ set to $\log_2 (1+\delta)$.
For the first term, we get:
\begin{align*}
\sum_{w = 1}^{\lambda m/\log n} \frac{\alpha^{-tw} n^{(1 - \beta t)w}}{w!}
  & \leq \sum_{w = 0}^{\infty} \frac{\alpha^{-tw} n^{(1 - \beta t)w}}{w!} - \frac{\alpha^{0} n^0}{0!}  \\
  & = \exp(\alpha^{-t} n^{1 - \beta t}) - 1
\end{align*}
Choose any $t > 1/\beta$. Then the second term converges to zero as well.
\end{proof}

\begin{lemma}
\label{lemma:asymp2}
Let $\alpha, \beta \in (0,1), m = \alpha n^\beta, c \in \mathbb{Z},$ and $q = 2^{m + c}$. Let $w^*$ and $\epsilon(n, m, q, f)$ be as in Definition~\ref{def:wstar-etc}. Then, for all $\gamma > 1$, $k > \frac{1-\beta}{2\beta}$ and $f = \frac{k\log^2{m}}{m}$, and values of $n$ greater than some $N > 0$,
\[ \epsilon(n, m, q, f) \leq \gamma \frac{2^c}{q-1} \]  
\end{lemma}

\begin{proof}
For any $0 < \delta < 1$, let $\lambda^* = \frac{\log_2(1+\delta)}{1-\beta}$, then $\forall \lambda < \lambda^*$, by Lemma~\ref{cor:wasymp2}, for all values of $n$ sufficiently large, 
\[ \frac{\lambda m}{\log n} \leq w^*(n) = w^* \]
We can write:
\begin{align*} 
  & (q-1) \epsilon(n, m, q, f) \\
  & = \sum_{w = 1}^{\lambda m/\log n} \binom{n}{w} \frac{1}{2^m} \left(1 + (1 - 2f)^w\right)^m + \\
    & \qquad \sum_{w = 1 + \lambda m/\log n}^{w^*} \binom{n}{w} \frac{1}{2^m} \left(1 + (1 - 2f)^w\right)^m + \\
    & \qquad \frac{r}{2^m} \left(1 + (1 - 2f)^{w^* + 1}\right)^m \\
  & \leq \sum_{w = 1}^{\lambda m/\log n} \binom{n}{w} \frac{1}{2^m} \left(1 + (1 - 2f)^w\right)^m + \\
    & \qquad \sum_{w = 1}^{w^*} \binom{n}{w} \frac{1}{2^m} \left(1 + (1 - 2f)^{\lambda m/\log n}\right)^m + \\
    & \qquad \frac{r}{2^m} \left(1 + (1 - 2f)^{\lambda m/\log n}\right)^m \\
  & = \sum_{w = 1}^{\lambda m/\log n} \binom{n}{w} \frac{1}{2^m} \left(1 + (1 - 2f)^w\right)^m + \\
    & \qquad \frac{q-1}{2^m} \left(1 + (1 - 2f)^{\lambda m/\log n}\right)^m \\
  & = A_n + B_n
\end{align*}

By Lemma~\ref{lemma:partial2}, from our choice of $\lambda$, for any $k > 0$,  $f = \frac{k\log^2{m}}{m}$, we have:
\[ \limn A_n = 0 \]
For $B_n$, when n is sufficiently large, we have:
\begin{align*}
B_n & = \frac{q-1}{2^m} \left(1 + \left(1 - 2\frac{k\log^2{m}}{m}\right)^{\lambda m/\log n}\right)^m \\
    & \leq 2^c \left(1 + \exp \left(- 2\frac{k\log^2{m}}{m} \frac{\lambda m}{\frac{1}{\beta}(\log{m} - \log{\alpha})}\right)\right)^m \\
    & \leq 2^c \left(1 + \exp \left(- 2\frac{k\log{m}}{m}\lambda \beta m\right)\right)^m \\
    & = 2^c \left(1 + m^{-2k\lambda\beta}\right)^m \\
    & \leq 2^c \exp \left(m^{-2k\lambda\beta} m\right)
      = 2^c \exp \left(m^{1 - 2k\lambda\beta}\right)
\end{align*}
for all $\lambda > 1$, if we choose $k$ such that $k > 1/(2 \lambda\beta)$, then $1 - 2k\lambda\beta < 0$ and we have
\[ \lim\sup_{n \to \infty} B_n \leq 2^c \]
Combining the two results, as long as we choose $k > 1/(2 \lambda\beta)$, we have, 
\[ \lim\sup_{n \to \infty} A_n + B_n \leq 2^c \]
which implies that for all $\gamma > 1$, for sufficiently large n:
\[ \epsilon(n, m, q, f) \leq \frac{A_n + B_n}{q-1} \leq \gamma \frac{2^c}{q-1} \]   

Since $\lambda < \log_2(1 + \delta) / (1 - \beta)$, we need $k$ to be larger than $(1-\beta) / (2 \beta \log_2(1+\delta))$. Since $\delta$ can be chosen arbitrarily from $(0,1)$, it suffices to have:
\[ k > \frac{1-\beta}{2\beta} \]
as for any such $k$, we can always find a $\delta$ close enough to 1 such that the above condition is satisfied.
\end{proof}

\begin{proof}[\textbf{Proof of Theorem~\ref{thm:asymp} (part 3)}]
This proof is almost exactly the same as that of Theorem~\ref{thm:asymp} (part 2), following as a direct consequence of Lemma~\ref{lemma:asymp2} along with Corollary~1 and Theorem~2 of \namecite{ermon2014low}.
\end{proof}

\end{document}